\pdfoutput=1
\documentclass[12pt]{article}
\usepackage{amsmath, amssymb, amsthm, latexsym,float,psfrag,epsfig, natbib}
\usepackage{multirow,booktabs,graphicx,color}
\pdfoutput=1
\usepackage{epstopdf}

\usepackage{geometry}
\allowdisplaybreaks[3]

\theoremstyle{plain}
\newtheorem{theorem}{Theorem}

\newtheorem{lemma}{Lemma}

\def\T{{ \mathrm{\scriptscriptstyle T} }}

\floatstyle{ruled}
\newfloat{algorithm}{tbhp}{loa}
\floatname{algorithm}{Algorithm}

\def\be{\begin{equation}}
\def\ee{\end{equation}}
\def\ben{\begin{equation*}}
\def\een{\end{equation*}}
\def\bea{\begin{eqnarray}}
\def\eea{\end{eqnarray}}
\def\bean{\begin{eqnarray*}}
\def\eean{\end{eqnarray*}}

\def\sgn{\mathrm{sign}}

\def\bv{\boldsymbol{v}}

\def\bX{\boldsymbol{X}}

\def\bZ{\boldsymbol{Z}}

\def\bB{\boldsymbol{B}}

\def\bI{\boldsymbol{I}}
\def\bS{\boldsymbol{S}}

\def\bzero{\boldsymbol{0}}

\def\bSigma{\boldsymbol{\Sigma}}

\def\bDelta{\boldsymbol{\Delta}}
\def\bTheta{\boldsymbol{\Theta}}
\def\bLambda{\boldsymbol{\Lambda}}

\begin{document}

\title{Positive Definite $\ell_1$ Penalized Estimation of Large Covariance Matrices}
\author{Lingzhou Xue, Shiqian Ma and Hui Zou\\University of Minnesota}
\date{December 13, 2011\\ Accepted by JASA, August 2012}
\maketitle

\begin{abstract}
The thresholding covariance estimator has nice asymptotic properties for estimating sparse large covariance matrices, but it often has negative eigenvalues when used in real data analysis. To simultaneously achieve sparsity and positive definiteness, we develop a positive definite $\ell_1$-penalized covariance estimator for estimating sparse large covariance matrices. An efficient alternating direction method is derived to solve the challenging optimization problem and its convergence properties are established. Under weak regularity conditions, non-asymptotic statistical theory is also established for the proposed estimator. The competitive finite-sample performance of our proposal is demonstrated by both simulation and real applications.
\end{abstract}

\noindent {\textbf{Keywords}: Alternating direction methods;
Large covariance matrices;  Matrix norm; Positive-definite estimation; Sparsity; Soft-thresholding.

\section{Introduction}

Estimating covariance matrices is of fundamental importance for an abundance of statistical methodologies.
Nowadays, the advance of new technologies has brought massive high-dimensional data into various research fields, such as fMRI imaging,
web mining, bioinformatics, climate studies and risk management, and so on.
The usual sample covariance matrix is optimal in the classical setting with large samples and fixed low dimensions \citep{anderson1984}, but it
performs very poorly in the high-dimensional setting \citep*{johnstone2001}. In the recent literature, regularization
techniques have been used to improve the sample covariance matrix estimator, including banding \citep*{wu2003,bickel2008a},
tapering \citep*{furrer2007,czz2010} and thresholding \citep*{bickel2008b,karoui2008,rothman2009}.
Banding or tapering is very useful when the variables have a natural ordering and off-diagonal entries of the target
covariance matrix decays to zero as they move away from the diagonal.
On the other hand, thresholding is proposed for estimating permutation-invariant covariance matrices. Thresholding can be used to produce
consistent covariance matrix estimators when the true covariance matrix is bandable \citep{bickel2008b,cz2011}.
In this sense, thresholding is more robust than banding/tapering for real applications.

Let $\hat\bSigma_n=(\hat\sigma_{ij})_{1\le i,j\le p}$ be the sample covariance matrix. \cite*{rothman2009} defined the general thresholding covariance matrix estimator as
$
\hat\bSigma_{thr}=\{s_{\lambda}(\hat \sigma_{ij})\}_{1\le i,j\le p},
$
where $s_{\lambda}(z)$ is the generalized thresholding function. The generalized
thresholding function covers a number of commonly used shrinkage procedures, e.g.
the hard thresholding $s_{\lambda}(z)=zI_{\{|z|>\lambda\}}$,
the soft thresholding $s_{\lambda}(z)=\textrm{sign}(z)(|z|-\lambda)_+$,
the smoothly clipped absolute deviation thresholding \citep{scad}
and the adaptive lasso thresholding \citep{zou2006}.
Consistency results and explicit rates of convergence have been obtained for these regularized estimators
in the literature, e.g. \cite*{bickel2008a,bickel2008b}, \cite{karoui2008}, \cite{rothman2009}, \cite{cai2011}.
The recent work by \cite{cz2011} has established the minimax rate of convergence
under the $\ell_1$ matrix norm over a fairly wide range of classes of large covariance matrices, where
the thresholding estimator is shown to be minimax rate optimal. The existing theoretical and empirical results show
no clear favoritism to a particular thresholding rule. In this paper we focus on the soft-thresholding because it can be formulated
as the solution of a convex optimization problem. Let $\|\cdot\|_F$ be the Frobenius norm and $|\cdot|_1$ be the element-wise $\ell_1$-norm of all non-diagonal elements. Then the soft-thresholding covariance estimator is equal to
\be\label{soft-thresholding}
\hat\bSigma=\arg\min_{\bSigma}~\frac{1}2\|\bSigma-\hat\bSigma_n\|_F^2+\lambda|\bSigma|_1.
\ee

However, there is no guarantee that the thresholding estimator is always positive definite. Although the positive definite property is guaranteed in the asymptotic setting with high probability, the actual estimator can be an indefinite matrix, especially in real data analysis. To illustrate this issue, we consider the Michigan lung cancer gene-expression data \citep{beer2002} which have $86$ tumor samples from patients with lung adenocarcinomas and $5217$ gene expression values for each sample. More details about this dataset are referred to \cite{beer2002} and \cite{gsea2005}. We randomly choose $p$ genes ($p=200, 500$), and obtain the soft-thresholding sample correlation matrix for these genes. We repeat the process ten times for $p=200$ and $500$ respectively, and each time the thresholding parameter $\lambda$ is selected via the 5-fold cross validation. We found that none of the soft-thresholding estimators would become positive definite for both $p=200$ and $500$. On average, there exist $22$ and $124$ negative eigenvalues for the soft-thresholding estimator for $p=200$ and $p=500$, respectively. Figure \ref{plot:michigan_data} displays the $30$ smallest eigenvalues for $p=200$ and the $130$ smallest eigenvalues for $p=500$.

\begin{figure}[!ht]
\centering
\small{
\begin{minipage}[t]{0.4\linewidth} \centering
\includegraphics[width=\textwidth]{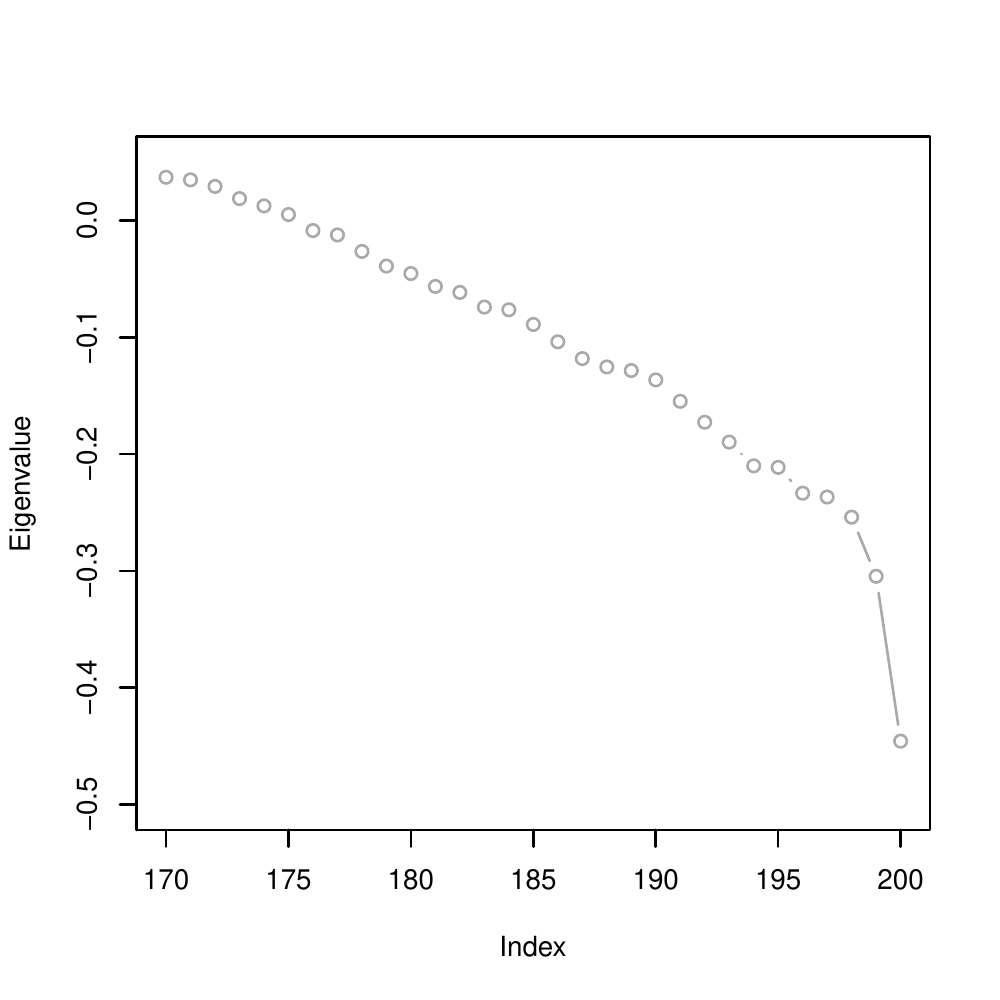}
{\small (A) $p=200$: the minimal $30$ eigenvalues}
\end{minipage}
\begin{minipage}[t]{0.4\linewidth} \centering
\includegraphics[width=\textwidth]{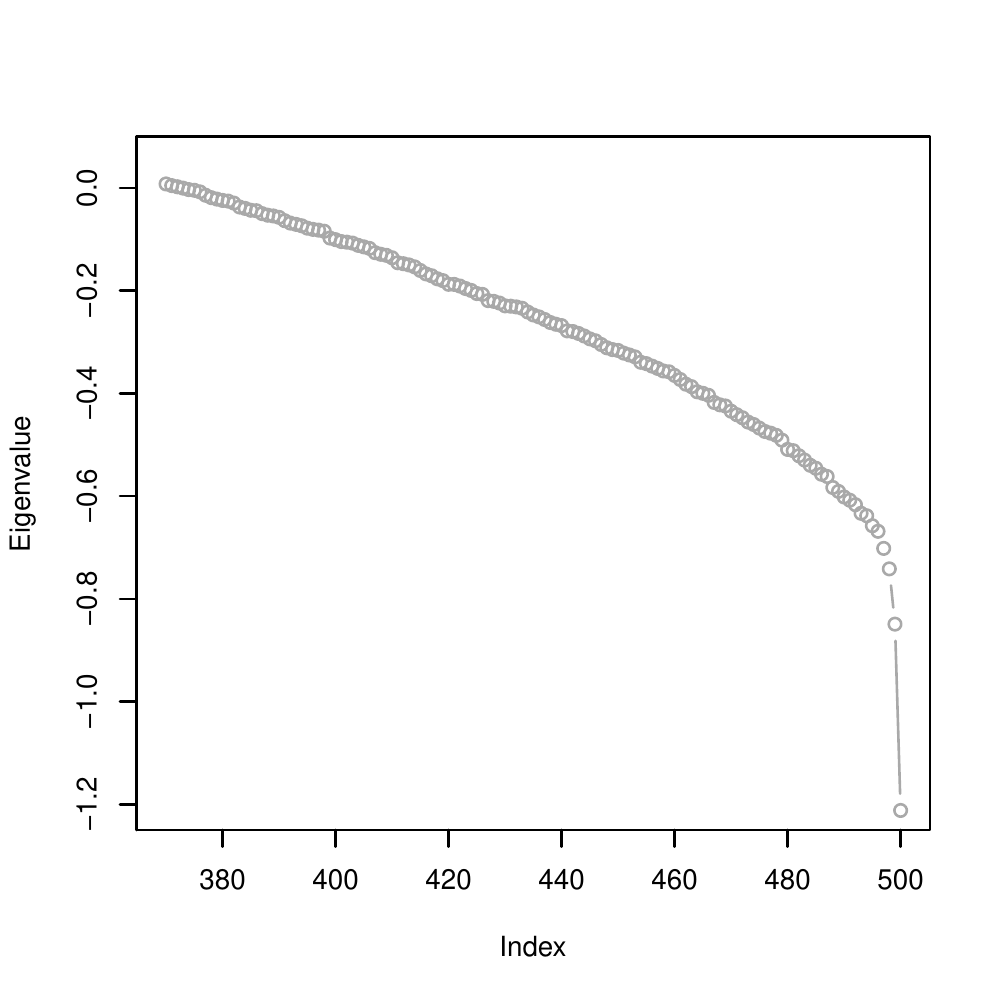}
{\small (B) $p=500$: the minimal $130$ eigenvalues}
\end{minipage}
\caption{Illustration of the indefinite soft-thresholding estimator in the Michigan lung cancer data}
}\label{plot:michigan_data}
\end{figure}

To deal with the indefiniteness, one possible solution is to utilize the eigen-decomposition of $\hat\bSigma$, and project $\hat\bSigma$ into the convex cone $\{\bSigma\succeq 0\}$. Assume that $\hat\bSigma$ has the eigen-decomposition $\hat\bSigma=\sum_{i=1}^p\hat\lambda_i\bv_i^\T\bv_i$, and then a positive semidefinite estimator $\tilde\bSigma^+$ can be obtained by setting $\tilde\bSigma^+=\sum_{i=1}^p \max(\hat\lambda_i,0)\bv_i^\T\bv_i$. However, this strategy does not work well for sparse covariance matrix estimation, because the projection destroys the sparsity pattern of $\hat\bSigma$. Consider the Michigan data again. After semidefinite projection, the soft-thresholding estimator has no zero entry.

In order to simultaneously achieve sparsity and positive semidefiniteness, a natural solution is to add the positive semidefinite constraint to (\ref{soft-thresholding}). Consider the following constrained $\ell_1$ penalization problem
\be\label{soft-thresholding.plus0}
\hat\bSigma^{+}=\arg\min_{\bSigma\succeq 0}~\|\bSigma-\hat\bSigma_n\|_F^2/2+\lambda|\bSigma|_1.
\ee
Note that the solution to (\ref{soft-thresholding.plus0}) could be positive semidefinite. To obtain a positive definite covariance estimator, we can consider the positive definite constraint $\{\bSigma\succeq \epsilon\bI\}$ for some arbitrarily small $\epsilon>0$. Then the modified $\hat \bSigma^{+}$ is always positive definite. In this work, we focus on solving the positive definite $\hat \bSigma^{+}$ as follows
\be\label{soft-thresholding.plus}
\hat\bSigma^{+}=\arg\min_{\bSigma\succeq \epsilon\bI}~\|\bSigma-\hat\bSigma_n\|_F^2/2+\lambda|\bSigma|_1.
\ee

Despite its natural motivation, (\ref{soft-thresholding.plus}) is actually a very challenging optimization problem due to the positive semidefinite constraint. To our best knowledge, the first attempt for solving (\ref{soft-thresholding.plus}) was recently proposed by \cite{rothman2011} who added the log-determinant barrier function to (\ref{soft-thresholding.plus}):
\be\label{rothman}
\breve\bSigma^{+}=\arg\min_{\bSigma\succ 0}~\|\bSigma-\hat\bSigma_n\|_F^2/2-\tau\log \det(\bSigma)+\lambda|\bSigma|_1,
\ee
where the barrier parameter $\tau$ is a small positive constant, say $10^{-4}$.
From the optimization viewpoint, (\ref{rothman}) is similar to the graphical lasso criterion \citep{glasso} which also has a log-determinant part and the element-wise $\ell_1$-penalty.
\cite{rothman2011} derived an iterative procedure to solve (\ref{rothman}) . \cite{rothman2011}'s proposal is based on heuristic arguments and its convergence property is unknown.

In this paper we present an alternating direction algorithm for solving (\ref{soft-thresholding.plus}) directly. Numerical examples show that our algorithm is much faster than the log-barrier method.
We further prove the convergence properties of our algorithm and discuss the statistical properties of the positive-definite constrained $\ell_1$ penalized covariance estimator.

\section{Alternating Direction Algorithm}
We use an alternating direction method to solve (\ref{soft-thresholding.plus}) directly.
The alternating direction method is closely related to the operator-splitting method that has a long history back to 1950s for solving numerical partial differential equations, see e.g., \cite{Douglas-Rachford-56,Peaceman-Rachford-55}. Recently, the alternating direction method has been revisited and successfully applied to solving large scale problems arising from different applications. For example, \cite*{ma2010nips} introduced the alternating linearization methods to efficiently solve the graphical lasso optimization problem. We refer to \cite{fortin1983,glowinski1989} for more details on operator-splitting and alternating direction methods.

In the sequel, we propose an alternating direction method to solve the $\ell_1$ penalized covariance matrix estimation problem (\ref{soft-thresholding.plus}) under the positive-semidefinite constraint. We first introduce a new variable $\bTheta$ and an equality constraint as follows
\be\label{soft-thresholding.alm}
(\hat\bTheta^{+},\hat\bSigma^{+})=\arg\min_{\bTheta,\bSigma}~\{\|\bSigma-\hat\bSigma_n\|_F^2/2+\lambda|\bSigma|_1:
~\bSigma=\bTheta,~\bTheta\succeq \epsilon\bI\}.
\ee
The solution to (\ref{soft-thresholding.alm}) gives the solution to (\ref{soft-thresholding.plus}).
To deal with the equality constraint in (\ref{soft-thresholding.alm}), we shall minimize its augmented Lagrangian function
for some given penalty parameter $\mu$, i.e.
\be\label{AugLagFunction}
L(\bTheta,\bSigma;\bLambda)
=
\|\bSigma-\hat\bSigma_n\|_F^2/2+\lambda|\bSigma|_1
-\langle \bLambda, \bTheta-\bSigma\rangle+\|\bTheta-\bSigma\|_F^2/(2\mu),
\ee where $\bLambda$ is the Lagrange multiplier.
We iteratively solve \be \label{subproblem1}
(\bTheta^{i+1},\bSigma^{i+1})=\arg\min L(\bTheta,\bSigma;\bLambda^i)\ee and then
update the Lagrangian multiplier $\bLambda^{i+1}$ by
$$
\bLambda^{i+1}=\bLambda^i-(\bTheta^{i+1}-\bSigma^{i+1})/\mu.
$$
For (\ref{subproblem1}) we do it by alternatingly minimizing
$L(\bTheta,\bSigma;\bLambda^i)$ with respect to $\bTheta$ and $\bSigma$.

To sum up, the entire algorithm proceeds as follows:
\begin{description}
\item[] For $i=0,1,2,\ldots$, solve the following three sub-problems sequentially till convergence
\be
\bTheta \ \mathrm{step:} \quad \bTheta^{i+1} = \arg\min_{\bTheta\succeq \epsilon\bI} L(\bTheta,\bSigma^{i};\bLambda^i) \label{alg:ADM-1}
\ee
\be
\bSigma \ \mathrm{step:} \quad \bSigma^{i+1} = \arg\min_{\bSigma} L(\bTheta^{i+1},\bSigma;\bLambda^i) \label{alg:ADM-2}
\ee
\be
\bLambda \ \mathrm{step:} \quad \bLambda^{i+1}=\bLambda^i-(\bTheta^{i+1}-\bSigma^{i+1})/\mu. \label{alg:ADM-3}
\ee
\end{description}
To further simplify the alternating direction algorithm, we derive the closed-form solutions for (\ref{alg:ADM-1})--(\ref{alg:ADM-2}). Consider the $\bTheta$ step. Define $(\bZ)_{+}$
as the projection of a matrix $\bZ$ onto the convex cone $\{\bTheta\succeq \epsilon\bI\}$.
Assume that $\bZ$ has the eigen-decomposition $\sum_{i=1}^p\lambda_i\bv_i^\T\bv_i$, and then $(\bZ)_{+}$ can be obtained as $\sum_{i=1}^p \max(\lambda_i,\epsilon)\bv_i^\T\bv_i$. Then the $\bTheta$ step can be analytically solved as follows
\bean
\bTheta^{i+1}
&=& \arg\min_{\bTheta\succeq \epsilon\bI} L(\bTheta,\bSigma^{i};\bLambda^i) \\
&=& \arg\min_{\bTheta\succeq \epsilon\bI} -\langle\bLambda^i,\bTheta\rangle+\|\bTheta-\bSigma^{i}\|_F^2/(2\mu)\\
&=& (\bSigma^{i}+\mu\bLambda^i)_{+}.
\eean

Next, define an entry-wise soft-thresholding rule for all the non-diagonal elements of a matrix $\bZ$ as $\bS(\bZ,\tau)=\{s(z_{ij},\tau)\}_{1\le i,j\le p}$  with
\[
s(z_{ij},\tau) = \sgn(z_{ij})\max(|z_{ij}|-\tau,0)I_{\{i\neq j\}}+z_{ij}I_{\{i=j\}}.
\]
Then  the $\bSigma$ step  has  a closed-form solution given below
\bean
\bSigma^{i+1}
&=& \arg\min_{\bSigma} L(\bTheta^{i+1},\bSigma;\bLambda^i) \\
&=& \arg\min_{\bSigma} \|\bSigma-\hat\bSigma_n\|_F^2/2+\lambda|\bSigma|_1
+\langle\bLambda^i,\bSigma\rangle+\|\bSigma-\bTheta^{i+1}\|_F^2/(2\mu)\\
&=& \{\bS(\mu(\hat\bSigma_n-\bLambda^i)+\bTheta^{i+1}, \lambda\mu)\}/(1+\mu).
\eean
Algorithm 1 shows the complete details of our alternating direction method for (\ref{soft-thresholding.plus}).  In Section 4 we provide the convergence analysis of Algorithm 1 and prove that Algorithm 1 always converges to the optimal solution of \eqref{soft-thresholding.alm} from any starting point.

\begin{algorithm}\label{alg:ADM-1-formal}
\caption{Our alternating direction method for the $\ell_1$ penalized covariance estimator}
\begin{enumerate}
\item Input: $\mu$, $\bSigma^0$ and $\bLambda^0$.

\item Iterative alternating direction augmented Lagrangian step: for the $i$-th iteration
\begin{enumerate}
\item [2.1] Solve $\bTheta^{i+1} = (\bSigma^{i}+\mu\bLambda^i)_{+}$;
\item [2.2] Solve $\bSigma^{i+1} = \{\bS(\mu(\hat\bSigma_n-\bLambda^i)+\bTheta^{i+1}, \lambda\mu)\}/(1+\mu)$;
\item [2.3] update $\bLambda^{i+1}=\bLambda^i-(\bTheta^{i+1}-\bSigma^{i+1})/\mu$.
\end{enumerate}
\item Repeat the above cycle till convergence.
\end{enumerate}
\end{algorithm}

In our implementation we use the soft-thresholding estimator as the initial value for both $\bTheta^0$ and $\bSigma^0$, and we set $\bLambda^0$ as a zero matrix. The value for $\mu$ is 2. Before invoking Algorithm 1, we always check whether the soft-thresholding estimator is positive definite. If yes, then the soft-threhsolding estimator is the final solution to (\ref{soft-thresholding.plus}).

\section{Numerical Examples}
\subsection{Simulation}
Before delving  into theoretical analysis of the algorithm and the resulting estimator, we first use simulation to show the competitive performance of our proposal.
In all examples we standardize the variables to have zero mean and unit variance.
In each simulation model, we generated $100$ independent datasets, each with $n=50$ independent $p$-variate random vectors from the multivariate normal distribution with mean $0$ and covariance matrix $\bSigma_0=(\sigma^0_{ij})_{1\le i,j\le p}$ for $p=100, 200~\&~500$. We considered two covariance models with different sparsity patterns:
\begin{description}
  \item[Model 1:] $\sigma^0_{ij}=(1-{|i-j|}/10)_+$.
  \item[Model 2:] partition the indices $\{1,2,\ldots,p\}$ into $K=p/20$ non-overlapping subsets of equal size, and let $i_k$ denote the maximum index in $I_k$.
      \[
      \sigma^0_{ij}=0.6I_{\{i=j\}} + 0.4\sum_{k=1}^{K}I_{\{i\in I_k, j\in I_k\}}+0.4\sum_{k=1}^{K-1}(I_{\{i=i_k,j\in I_{k+1}\}}+I_{\{i\in I_{k+1}, j=i_k\}})
      \]
\end{description}
Model 1 has been used in \cite{bickel2008a} and \cite{cai2011}, and Model 2 is similar to the overlapping block diagonal design used in \cite{rothman2011}.

First, we compare the run times of our estimator $\hat\bSigma^+$ with the log-barrier estimator $\breve\bSigma^+$ by \cite{rothman2011}. As shown in Table \ref{sim:timing}, our method is much faster than the log-barrier method.

\begin{table}[!ht]
\small{
\begin{center}
\caption{Total time (in seconds) for computing a solution path with 99 thresholding parameters $\lambda=\{0.01, 0.02, \cdots, 0.99\}$. Timing was carried out on an {\tt AMD} 2.8GHz processor.}\label{sim:timing}
\vspace{0.1in}
\begin{tabular}{c|ccc|ccc}
\toprule
& \multicolumn{3}{c|}{Model 1}	 &  \multicolumn{3}{c}{Model 2} \\
$p$ & 100 & 200 & 500 & 100 & 200 & 500\\
\hline
Our method  & 9.2 &  65.2  & 1156.0  & 7.5 &  51.1  & 986.6\\
Rothman's method& 84.1 & 822.1  & 35911.8   & 51.3 &  611.1  & 32803.0\\
\bottomrule
\end{tabular}
\end{center}
}
\end{table}

In what follows, we compare the performance of $\hat\bSigma^+$, $\breve\bSigma^+$  and  the soft-thresholding estimator $\hat\bSigma$. For all three regularized estimators, the thresholding parameter was chosen by $5$-fold cross-validation \citep{bickel2008b,rothman2009,cai2011}. The estimation performance is measured by the average losses under both the Frobenius norm and the spectral norm. The selection performance is examined by the false positive rate \[\frac{\#\{(i,j):~\hat\sigma_{ij}\neq 0~\&~\sigma_{ij}= 0\}}{\#\{(i,j):~\sigma_{ij}= 0\}}\] and the true positive rate \[\frac{\#\{(i,j):~\hat\sigma_{ij}\neq 0~\&~\sigma_{ij}\neq 0\}}{\#\{(i,j):~\sigma_{ij}\neq 0\}}.\] Moreover, we compare the average number of negative eigenvalues and the percentage of positive-definiteness to check the positive-definiteness.

Table \ref{stcor:model1} and Table \ref{stcor:model2} show the average metrics over 100 replications. The soft-thresholding estimator $\hat\bSigma$ is positive definite in 19 or fewer out of 100 simulation runs, while $\hat\bSigma^+$ and $\breve\bSigma^+$ can always guarantee a positive-definite estimator. The larger the dimension, the less likely for the soft-thresholding estimator to be positive definite. In terms of estimation, both $\hat\bSigma^+$ and $\breve\bSigma^+$ are more accurate than $\hat\bSigma$. As for the selection performance, $\hat\bSigma^+$ and $\breve\bSigma^+$ achieve a slightly better true positive rate than $\hat\bSigma$. Overall, $\hat\bSigma^+$ is the best among all three regularized estimators.

\begin{table}[!ht]
\small{
\begin{center}
\caption{Comparison of the three regularized estimators for Model 1. Each metric is averaged over 100 replications with the standard error shown in the bracket. NA means that the results for $\breve\bSigma^+$ (Rothman's method) are not available due to the extremely long run times.}\label{stcor:model1}
\vspace{0.1in}
\begin{tabular}{c|cc|cc|cc}
\toprule
&Frobenius&Spectral&False   &True    &Negative   &Positive \\
&norm     &norm    &positive&positive&eigenvalues&definiteness\\
\hline
&\multicolumn{6}{c}{$p=100$}\\
\hline
soft &8.41	&4.02	&24.5	&87.6	&2.24	&\multirow{2}{*}{53/100}\\
thresholding &(0.06)&(0.04)	&(0.1) &(0.0)	&(0.14)	&\\
our &8.40 &4.02	&24.8	&87.8	&0.00	&\multirow{2}{*}{100/100}\\
method &(0.06)&(0.04)	&(0.1) &(0.0)	&(0.00)&\\
Rothman's &8.40	&4.02	&24.5	&87.7	&0.00	&\multirow{2}{*}{100/100}\\
method &(0.06)&(0.04)	&(0.1) &(0.0)	&(0.00)&\\
\hline
&\multicolumn{6}{c}{$p=200$}\\
\hline
soft &13.82	&4.70	&14.3	&83.2&3.74	&\multirow{2}{*}{23/100}\\
thresholding &(0.06)&(0.03)	&(0.4)&(0.3) &(0.22)	&\\
our &13.80	&4.69	&14.6	&83.5 &0.00	&\multirow{2}{*}{100/100}\\
method &(0.06)&(0.03)	&(0.4)&(0.3) &(0.00)&\\
Rothman's &13.81	&4.69	&14.6	&83.5 &0.00	&\multirow{2}{*}{100/100}\\
method &(0.05)&(0.03)	&(0.4)&(0.3) &(0.00)&\\
\hline
&\multicolumn{6}{c}{$p=500$}\\
\hline
soft &25.15	&5.28	&6.3	&78.1&4.64	&\multirow{2}{*}{7/100}\\
thresholding &(0.11)	&(0.04)	&(0.2) &(0.3) &(0.60)	&\\
our &25.10	&5.28	&6.5	&78.3&0.00	&\multirow{2}{*}{100/100}\\
method &(0.11)	&(0.04)	&(0.2) &(0.3) &(0.00)&\\
Rothman's &NA &NA &NA &NA &NA &\multirow{2}{*}{NA} \\
method &NA &NA &NA &NA &NA &\\
\bottomrule
\end{tabular}
\end{center}
}
\end{table}

\begin{table}[!ht]
\small{
\begin{center}
\caption{Comparison of the three regularized estimators for Model 2. Each metric is averaged over 100 replications with the standard error shown in the bracket. NA means that the results for $\breve\bSigma^+$ (Rothman's method) are not available due to the extremely long run times.}\label{stcor:model2}
\begin{tabular}{c|cc|cc|cc}
\toprule
&Frobenius&Spectral&False   &True    &Negative   &Positive \\
&norm     &norm    &positive&positive&eigenvalues&definiteness\\
\hline
&\multicolumn{6}{c}{$p=100$}\\
\hline
soft &9.81	&4.87	&29.5	&97.2	&1.54	&\multirow{2}{*}{19/100}\\
thresholding &(0.07)&(0.05)	&(0.0) &(0.0)	&(0.14)	&\\
our &9.78 &4.85	&30.2	&97.3	&0.00	&\multirow{2}{*}{100/100}\\
method &(0.07)&(0.05)	&(0.0) &(0.0)	&(0.00)&\\
Rothman's &9.78	&4.85	&30.0	&97.3	&0.00	&\multirow{2}{*}{100/100}\\
method &(0.07)&(0.05)	&(0.0) &(0.0)	&(0.00)&\\
\hline
&\multicolumn{6}{c}{$p=200$}\\
\hline
soft &15.95	&5.90	&17.1	&94.1&3.93	&\multirow{2}{*}{7/100}\\
thresholding &(0.12)&(0.06)	&(0.4)&(0.3) &(0.27)	&\\
our &15.81	&5.84	&18.8	&95.0 &0.00	&\multirow{2}{*}{100/100}\\
method &(0.12)&(0.06)	&(0.3)&(0.3) &(0.00)&\\
Rothman's &15.83	&5.85	&18.3	&94.6 &0.00	&\multirow{2}{*}{100/100}\\
method &(0.12)&(0.06)	&(0.4)&(0.3) &(0.00)&\\
\hline
&\multicolumn{6}{c}{$p=500$}\\
\hline
soft &29.46	&6.92	&7.6	&87.7&3.84	&\multirow{2}{*}{4/100}\\
thresholding &(0.18)	&(0.07)	&(0.1) &(0.5) &(0.78)	&\\
our &29.17	&6.84	&8.7	&88.8&0.00	&\multirow{2}{*}{100/100}\\
method &(0.20)	&(0.06)	&(0.2) &(0.6) &(0.00)&\\
Rothman's &NA &NA &NA &NA &NA &\multirow{2}{*}{NA} \\
method &NA &NA &NA &NA &NA &\\
\bottomrule
\end{tabular}
\end{center}
}
\end{table}

\subsection{Real data}

 To demonstrate our proposal we further consider two gene expression datasets: one from a small round blue-cell tumors microarray experiment \citep{khan2001} and the other one from a cardiovascular microarray study \citep{efron2009,efron2010}. The first dataset has 64 training tissue samples with four types of tumors (23 EWS, 8 BL-NHL, 12 NB, and 21 RMS), and 6567 gene expression values for each sample. We applied the pre-filtering step used in Khan et al. (2001) and then picked the top 40 and bottom 160 genes based on the F-statistic as done in \cite{rothman2009}. The second dataset has 63 subjects with 44 healthy controls and 19 cardiovascular patients, and 20426 genes measured for each subject. We used the F-statistic to pick the top 50 and bottom 150 genes. By doing so, it is expected that there is weak dependence between the top and the bottom genes. We considered the soft-thresholding estimator \citep{bickel2008b}, the log-barrier estimator \citep{rothman2011} and our estimator. For all three estimators, the thresholding parameter was chosen by 5-fold cross validation.

\begin{figure}[!ht]
\centering
\begin{minipage}[t]{0.35\linewidth} \centering
    \includegraphics[width=\textwidth]{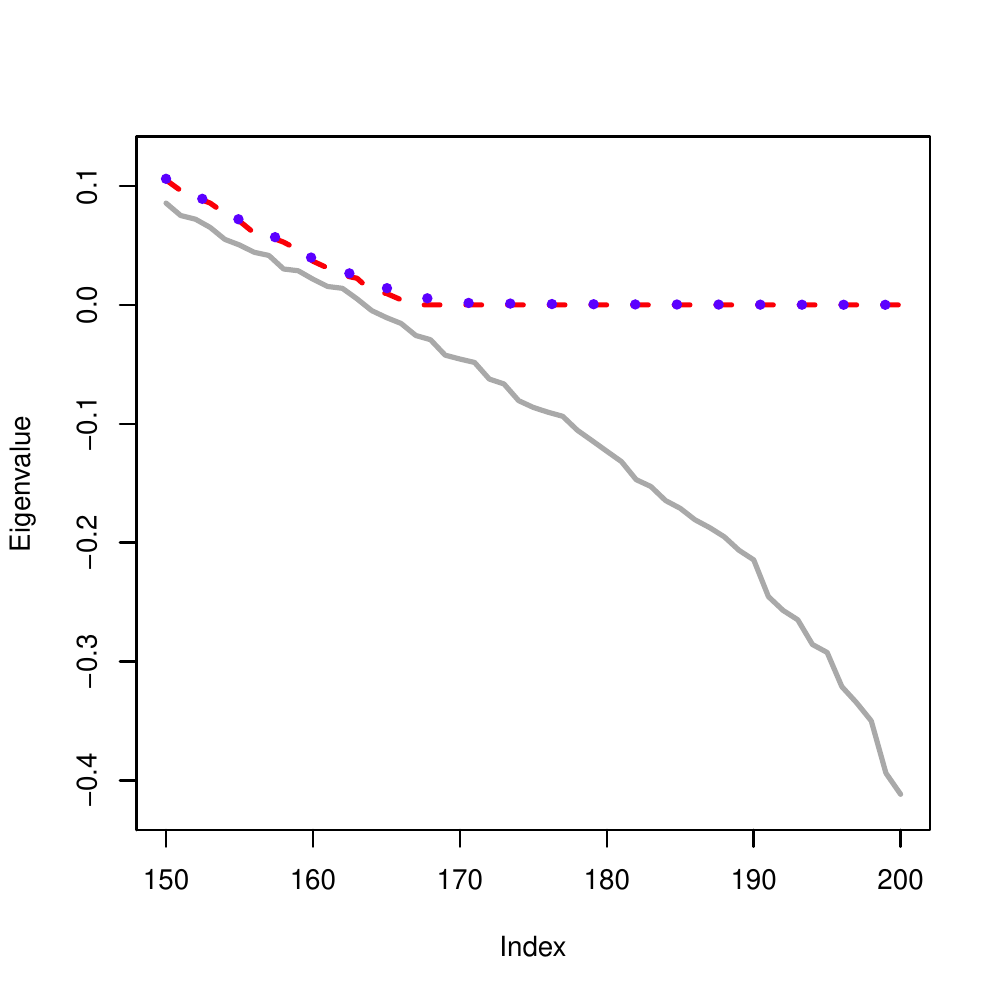}
    {\small (A)}
\end{minipage}
\begin{minipage}[t]{0.35\linewidth} \centering
    \includegraphics[width=\textwidth]{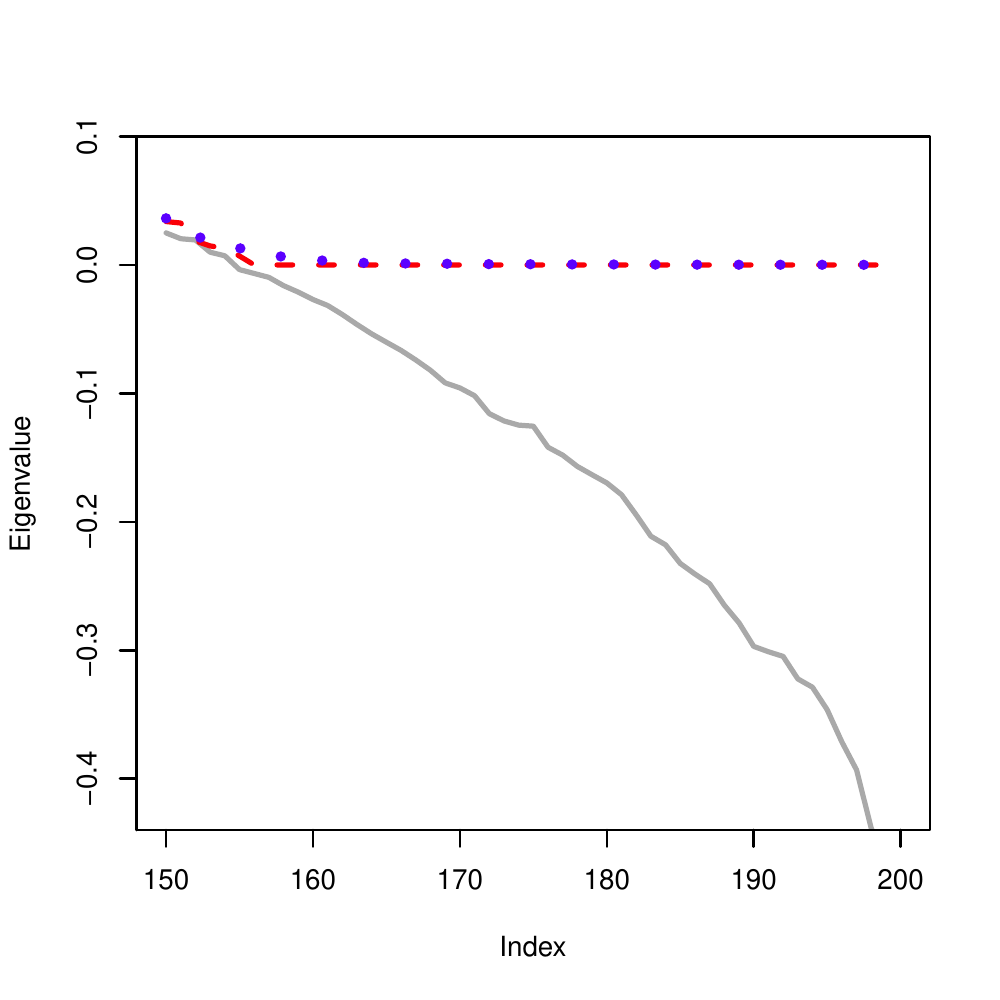}
    {\small (B)}
\end{minipage}
\caption{Plots of the bottom $50$ eigenvalues of all three regularized estimators for the small round blue-cell data (A) and the cardiovascular data (B): $\hat\bSigma$ (solid), $\hat\bSigma^+$ (dashed) and $\breve\bSigma^+$ (dotted).}\label{real:eigen}
\end{figure}

As evidenced in Plot \ref{real:eigen}, the soft-thresholding estimator yields an indefinite matrix for both real examples whereas the other two regularized estimators guarantee the positive-definiteness. The soft-thresholding estimator contains $37$ negative eigenvalues in the small round blue-cell data, and $46$ negative eigenvalues in the cardiovascular data.
Regularized correlation matrix estimation has a natural application in clustering when the dissimilarity measure is constructed using the correlation among features.
For both datasets we did hierarchical clustering using the three regularized estimators. The heat maps are shown in
Figure \ref{real:heatmap} in which the estimated sparsity pattern well matches the expected sparsity pattern.

\begin{figure}[!ht]
\centering
\begin{minipage}[t]{0.32\linewidth} \centering
    \includegraphics[width=\textwidth]{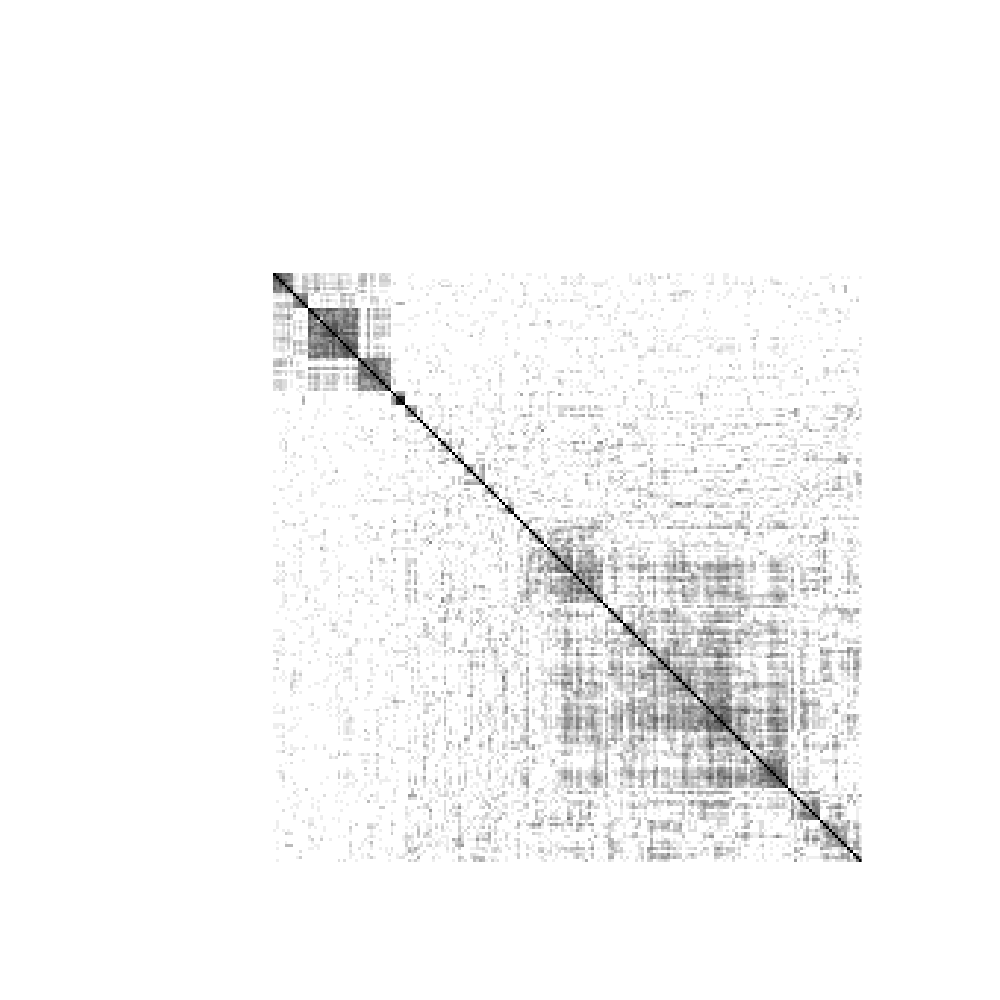}
    {\small (A1)}
\end{minipage}
\begin{minipage}[t]{0.32\linewidth} \centering
    \includegraphics[width=\textwidth]{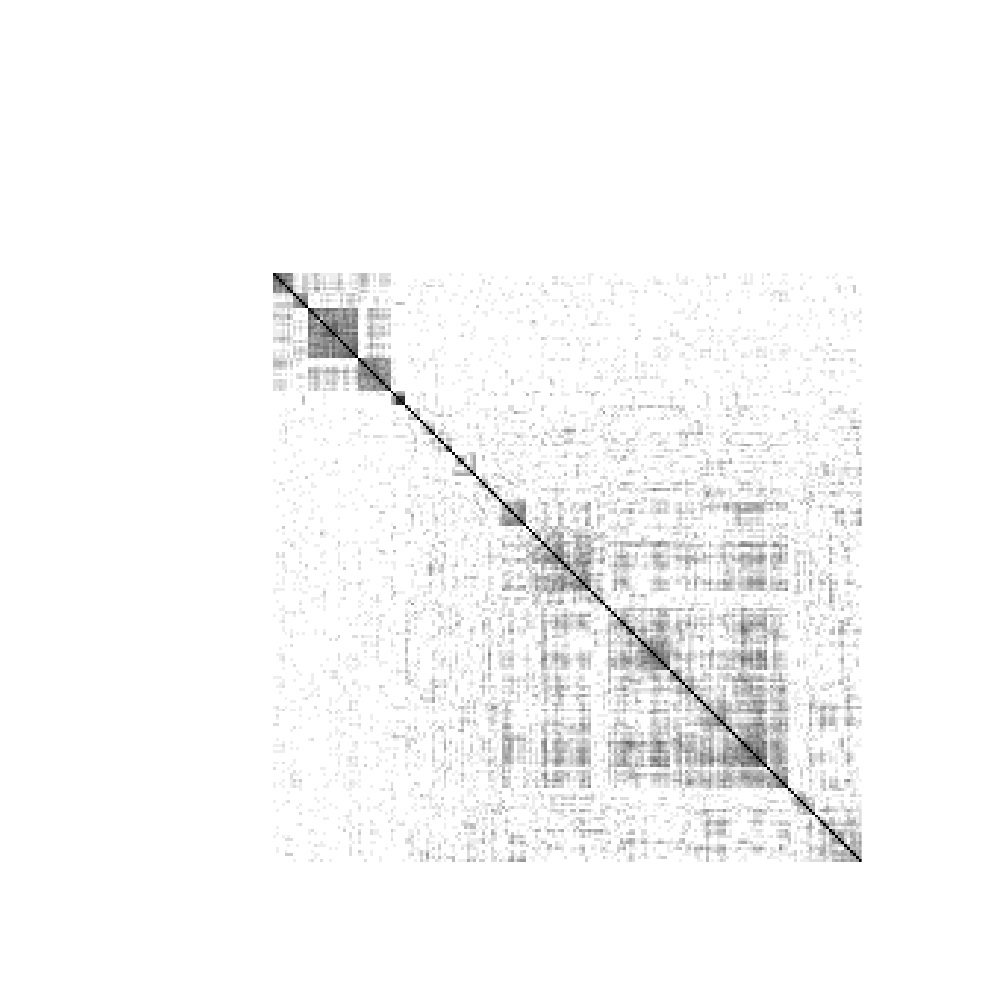}
    {\small (A2)}
\end{minipage}
\begin{minipage}[t]{0.32\linewidth} \centering
    \includegraphics[width=\textwidth]{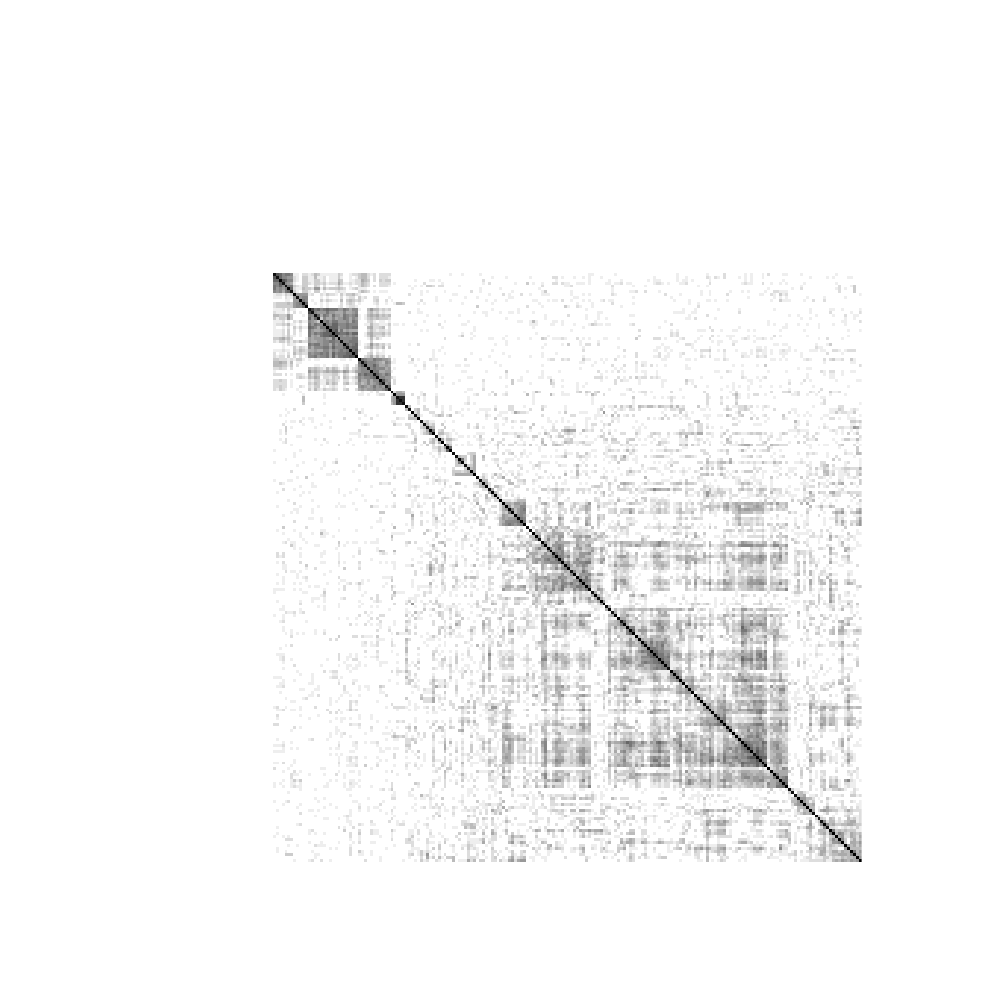}
    {\small (A3)}
\end{minipage}
\begin{minipage}[t]{0.32\linewidth} \centering
    \includegraphics[width=\textwidth]{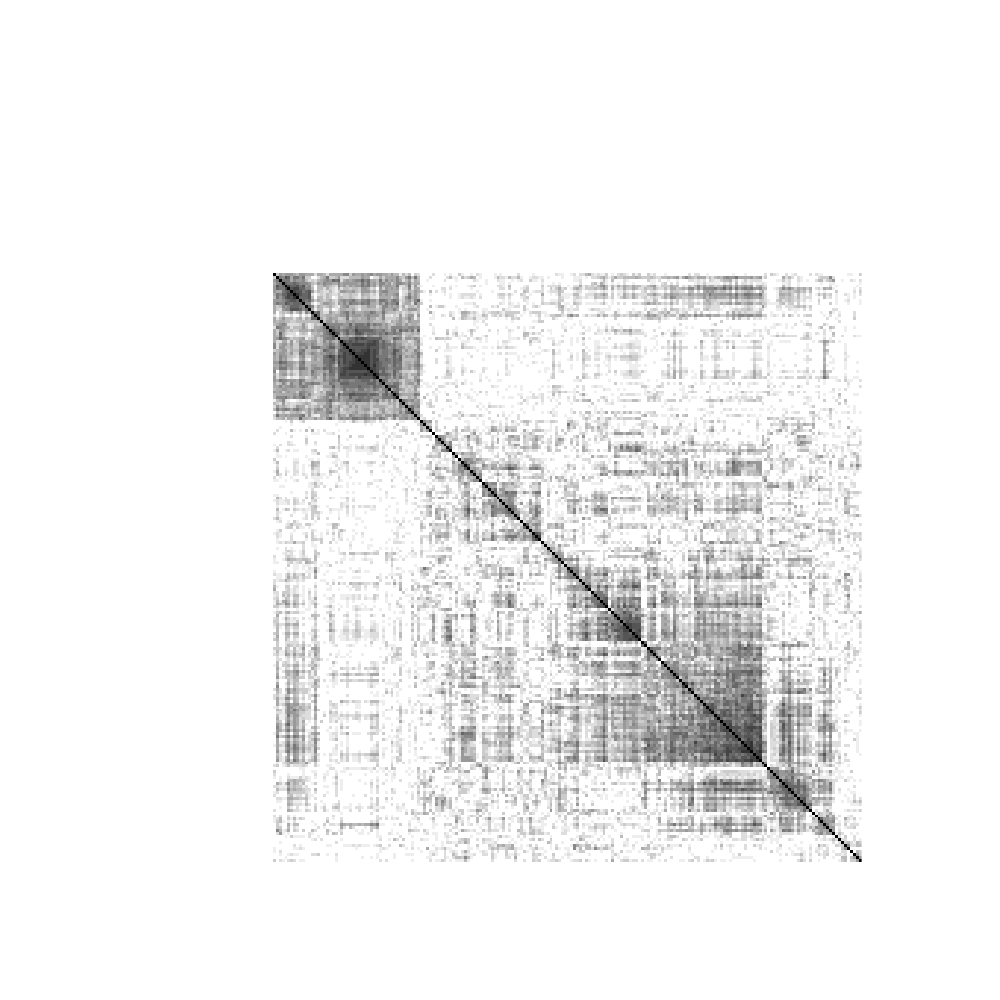}
    {\small (B1)}
\end{minipage}
\begin{minipage}[t]{0.32\linewidth} \centering
    \includegraphics[width=\textwidth]{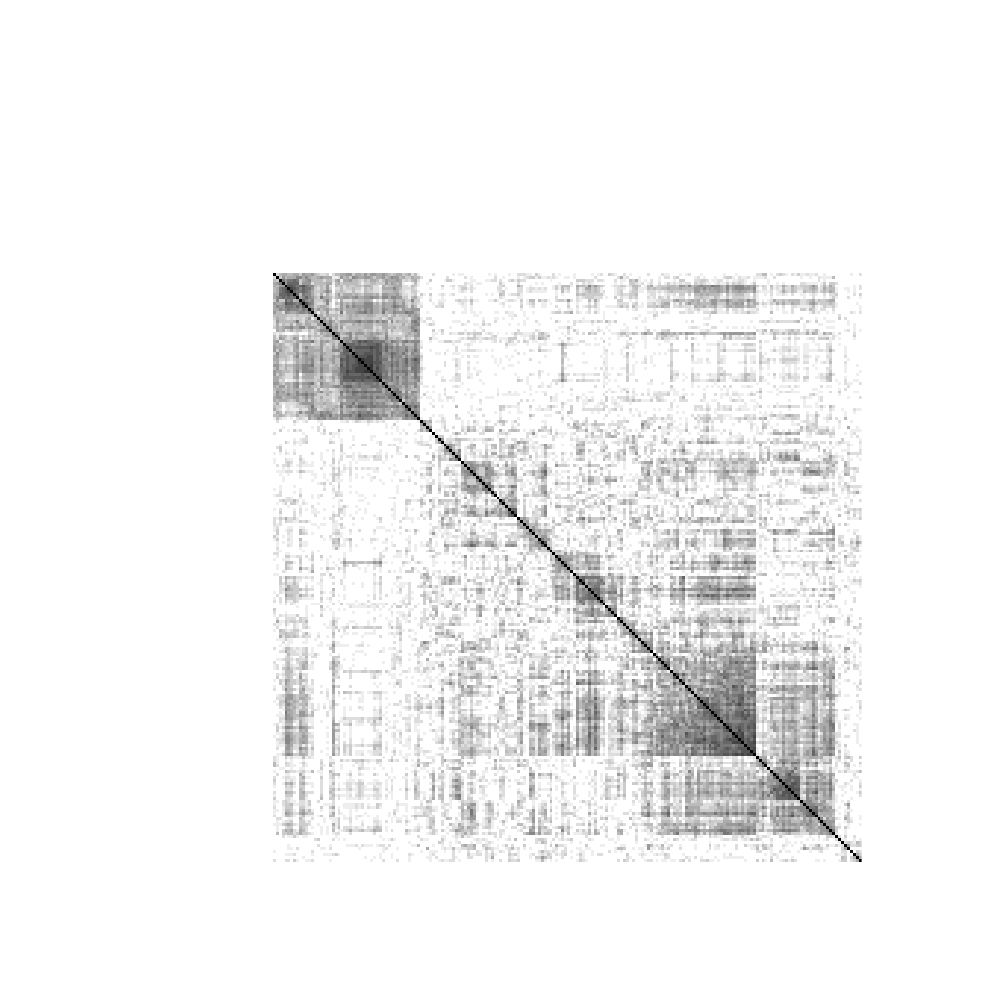}
    {\small (B2)}
\end{minipage}
\begin{minipage}[t]{0.32\linewidth} \centering
    \includegraphics[width=\textwidth]{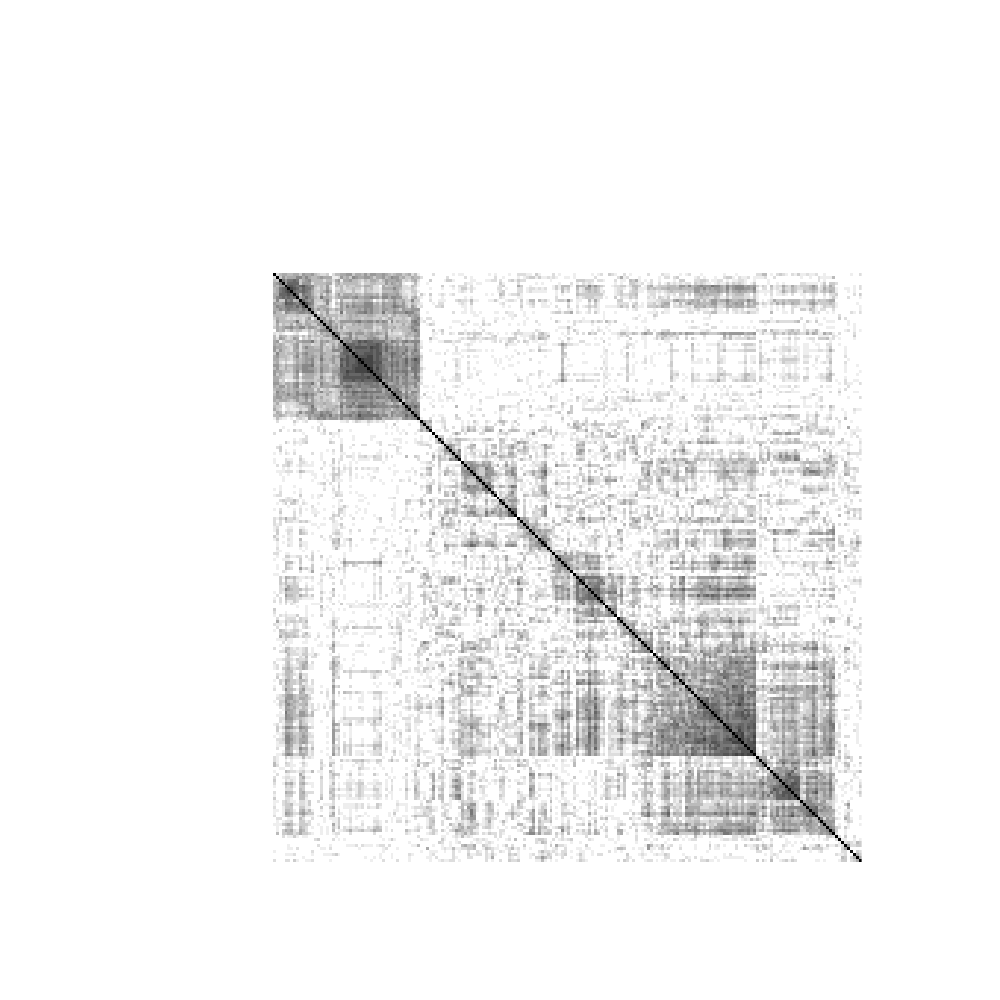}
    {\small (B3)}
\end{minipage}
\caption{Heat maps of the absolute values of three regularized sample correlation matrix estimator for the small round blue-cell data (A) and the cardiovascular data (B): $\hat\bSigma$ (A1, B1), $\hat\bSigma^+$ (A2, B2) and $\breve\bSigma^+$ (A3, B3). The genes are ordered by hierarchical clustering using the estimated correlations.}\label{real:heatmap}
\end{figure}

\begin{table}[!ht]
\small{
\begin{center}
\caption{Total time (in seconds) for computing a solution path with $99$ thesholding parameters. Timing was carried out on an {\tt AMD} 2.8GHz processor.}\label{real:timing}
\vspace{0.1in}
\begin{tabular}{c|cc}
\toprule
  & Blue cell data & Cardiovascular data \\
\hline
Our method   &  74.7  & 66.3 \\
Rothman's method &  1302.7  & 1575.3  \\
\bottomrule
\end{tabular}
\end{center}
}
\end{table}

Finally, we compared the average run times over $5$ cross validations for both $\hat\bSigma^+$  and $\breve\bSigma^+$, as shown in Table \ref{real:timing}. It is obvious that our proposal is much more efficient.

\section{Theoretical properties}

\subsection{Convergence analysis of the algorithm}

In this section, we prove that the sequence $(\bTheta^i,\bSigma^i,\bLambda^i)$ produced by the alternating direction method (Algorithm 1) converges to $(\hat\bTheta^+,\hat\bSigma^+,\hat\bLambda^+)$, where $(\hat\bTheta^+,\hat\bSigma^+)$ is an optimal solution of \eqref{soft-thresholding.alm} and $\hat\bLambda^+$ is the optimal dual variable. This automatically implies that Algorithm 1 gives an optimal solution of \eqref{soft-thresholding.plus}.

We define some necessary notation for ease of presentation. Let $G$ be a $2p$ by $2p$ matrix defined as
\[
G = \begin{pmatrix} \mu \bI_{p\times p} & 0 \\ 0 & (1/\mu)\bI_{p\times p} \end{pmatrix}.
\]
Define the norm $\|\cdot\|_G^2$ as $\|U\|_G^2 = \langle U, GU\rangle$ and the corresponding inner product $\langle \cdot,\cdot\rangle_G$ as $\langle U,V\rangle_G = \langle U,GV\rangle$. Before we give the main theorem about the global convergence of Algorithm 1, we need the following lemma.

\begin{lemma}\label{lem:Fejer-monotone}
Assume that $(\hat\bTheta^+,\hat\bSigma^+)$ is an optimal solution of \eqref{soft-thresholding.alm} and $\hat\bLambda^+$ is the corresponding optimal dual variable associated with the equality constraint $\bSigma = \bTheta$. Then the sequence $\{(\bTheta^i,\bSigma^i,\bLambda^i)\}$ produced by Algorithm 1 satisfies \be \label{lem:conclusion-eq} \|U^{i}-U^*\|_G^2 - \|U^{i+1}-U^*\|_G^2 \geq \|U^i - U^{i+1}\|_G^2, \ee where $U^* = (\hat\bLambda^+,~\hat\bSigma^+)^\T$ and $U^i = (\bLambda^i,~\bSigma^i )^\T$.
\end{lemma}

Now we are ready to give the main convergence result of Algorithm 1.
\begin{theorem}\label{the:main-convergence}
The sequence $\{(\bTheta^i,\bSigma^i,\bLambda^i)\}$ produced by Algorithm 1 from any starting point converges to an optimal solution of \eqref{soft-thresholding.alm}.
\end{theorem}


\subsection{Statistical analysis of the estimator}

Define $\bSigma^0$ as the true covariance matrix for the observations $\bX=(X_{ij})_{n\times p}$, and define the active set of $\bSigma^0=(\sigma^0_{jk})_{1\le j,k\le p}$ as $A_0=\{(j,k):\sigma^0_{jk}\neq 0 {,j\neq k}\}$ with the cardinality $s=|A_0|$. Denote by $\bB_{A_0}$ the Hadamard product $\bB_{p\times p}\circ(I_{\{(j,k)\in A_0\}})_{1\le j,k\le p}=(b_{jk}\cdot I_{\{(j,k)\in A_0\}})_{1\le j,k\le p}$. Define $\sigma_{\max}=\max_{j}\sigma^o_{jj}$ as the maximal
true variance in $\bSigma^0$.

\begin{theorem}\label{theorem:stat-convergence}
Assume that the true covariance matrix $\bSigma^0$ is positive definite.
\begin{itemize}
  \item [(i)] Under the exponential-tail condition that for all $|t|\le \eta$ and $1\le i\le n, 1\le j\le p$
    \[
    E\{\exp(tX_{ij}^2)\}\le K_1,
    \]
    we also assume that $\log p\le n$. For any $M>0$, we pick the thresholding parameter as \[
    \lambda=c_0^2\frac{\log p}n+c_1\left(\frac{\log p}n\right)^{1/2},\]
    where \[
    c_0=\frac{1}2eK_1\eta^{1/2}+\eta^{-1/2}(M+1)\] and \[c_1=2K_1(\eta^{-1}+\frac{1}4\eta\sigma_{\max}^2)\exp(\frac{1}2\eta\sigma_{\max})+2\eta^{-1}(M+2).\] With probability at least $1-3p^{-M}$, we have
    \[
    \|\hat\bSigma^+-\bSigma^0\|_F\le 5\lambda { (s+p)}^{1/2}.
    \]
  \item [(ii)] Under the polynomial-tail condition that for all $\gamma>0$, $\varepsilon>0$and $1\le i\le n, 1\le j\le p$
    \[
    E\{|X_{ij}|^{4(1+\gamma+\varepsilon)}\}\le K_2,
    \]
    we also assume that $p\le cn^\gamma$ for some $c>0$. For any $M>0$, we pick the thresholding parameter as \[
    \lambda=8(K_2+1)(M+1)\frac{\log p}n+8(K_2+1)(M+2)\left(\frac{\log p}n\right)^{1/2},\]
    With probability at least $1-O(p^{-M})-3K_2p(\log n)^{2(1+\gamma+\varepsilon)}n^{-\gamma-\varepsilon}$, we have
    \[
    \|\hat\bSigma^+-\bSigma^0\|_F\le 5\lambda { (s+p)}^{1/2}.
    \]
\end{itemize}
\end{theorem}

Define $d=\max_{j}\sum_{k}I_{\{\sigma_{jk}\neq 0\}}$ and assume that $\sigma_{\max}$ is bounded by a fixed constant, then we can pick $\lambda=O((\log p/n)^{1/2})$ to
achieve the minimax optimal rate of convergence under the Frobenius norm as in Theorem 4 of \cite{cz2012} that \ben
\frac{1}p\|\hat\bSigma^+-\bSigma^0\|_F^2=O_p\left((1+\frac{s}p)\frac{\log p}n\right)= O_p\left(d\frac{\log p}n\right). \een}
However, to attain the same rate in the presence of the log-determinant barrier term, \cite{rothman2011} instead would
require that $\sigma_{\textrm{min}}$, the minimal eigenvalue of the true covariance matrix, should be bounded away from zero by some positive constant,
and also that the barrier parameter should be bounded by some positive quantity. We would like to point out that if $\sigma_{\textrm{min}}$ is
bounded away from zero, then the soft-thresholding estimator $\hat\bSigma_{st}$ will be positive-definite with an overwhelming probability tending to $1$,
\citep{bickel2008b,cz2011,cz2012}. Therefore the theory requiring a lower bound on $\sigma_{\textrm{min}}$ is not very appealing.

\section{Conclusions}
The soft-thresholding estimator has been shown to enjoy good asymptotic properties for estimating large sparse covariance matrices.
But its positive definiteness property can be easily violated, which means the soft-thresholding estimator could be in principle an inadmissible estimator for covariance matrices.
In this paper we have put the soft-thresholding estimator in a convex optimization framework and considered a natural modification by imposing the positive definiteness constraint.
We have developed a fast alternating direction method to solve the constrained optimization problem and the resulting estimator retains the sparsity and positive definiteness properties simultaneously. The algorithm and the new estimator are supported by numerical and theoretical results.

\section*{Acknowledgement}
We thank Adam Rothman  for sharing his code. Shiqian Ma's research is supported by the National Science Foundation postdoctoral fellowship through Institute for Mathematics and Its Applications at University of Minnesota. Hui Zou's research is supported in part by grants from the National Science Foundation and the Office of Naval Research.

\section*{Appendix: Technical Proofs}

\begin{proof}[Proof of Lemma \ref{lem:Fejer-monotone}]
Since $(\hat\bTheta^+,\hat\bSigma^+,\hat\bLambda^+)$ is optimal to \eqref{soft-thresholding.alm}, it follows from the KKT conditions that the followings hold.
\be\label{KKT-1}
(-\hat\bLambda^+-\hat\bSigma^++\hat\bSigma_n)_{j\ell} /{\lambda}\in \partial |\hat\bSigma_{j\ell}^+|,\quad \forall j=1,\ldots,p, \ell=1,\ldots,p \mbox{ and } j\neq \ell, \ee
\be\label{KKT-1.5} (\hat\bSigma^+-\hat{\bSigma}_n)_{jj} + \hat\bLambda^+_{jj}=0, \quad \forall j=1,\ldots,p, \ee
\be\label{KKT-2}\hat\bTheta^+=\hat\bSigma^+,\ee \be\label{KKT-2.5} \hat{\bTheta}^+ \succeq \epsilon \bI, \ee and \be\label{KKT-3}\langle \hat{\bLambda}^+,\bTheta-\hat{\bTheta}^+ \rangle\leq 0,\quad\forall\bTheta\succeq\epsilon\bI.\ee

Note that the optimality conditions for the first subproblem in Algorithm 1, i.e. the subproblem with respect to $\bTheta$ in \eqref{alg:ADM-1}, are given by
\be\label{opt-cond-theta}\langle\bLambda^i-(\bTheta^{i+1}-\bSigma^i)/\mu,\bTheta-\bTheta^{i+1}\rangle\leq 0,\quad\forall\bTheta\succeq\epsilon\bI.\ee Using the updating formula for $\bLambda^i$ in Algorithm 1, i.e., \be\label{update-Lambda}\bLambda^{i+1}=\bLambda^i-(\bSigma^{i+1}-\bTheta^{i+1})/\mu,\ee
\eqref{opt-cond-theta} can be rewritten as
\be\label{opt-cond-theta-1} \langle\bLambda^{i+1}-(\bSigma^{i+1}-\bSigma^i)/\mu,\bTheta-\bTheta^{i+1}\rangle\leq 0,\quad\forall\bTheta\succeq\epsilon\bI.\ee Now by letting $\bTheta=\bTheta^{i+1}$ in \eqref{KKT-3} and $\bTheta=\hat{\bTheta}^+$ in \eqref{opt-cond-theta-1}, we can get that
\be\label{proof-lemma-theta-inequa-1}\langle\hat{\bLambda}^+,\bTheta^{i+1}-\hat{\bTheta}^+\rangle\leq 0, \ee and
\be\label{proof-lemma-theta-inequa-2}\langle\bLambda^{i+1}-(\bSigma^{i+1}-\bSigma^i)/\mu,\hat{\bTheta}^+-\bTheta^{i+1}\rangle \leq 0. \ee
Summing \eqref{proof-lemma-theta-inequa-1} and \eqref{proof-lemma-theta-inequa-2} yields
\be\label{proof-lemma-theta-inequa-3}\langle\bTheta^{i+1}-\hat\bTheta^+,
(\bLambda^{i+1}-\hat\bLambda^+)+(\bSigma^i-\bSigma^{i+1})/\mu\rangle \geq 0.\ee

The optimality conditions for the second subproblem in Algorithm 1, i.e., the subproblem with respect to $\bSigma$ in \eqref{alg:ADM-1} are given by \be \label{opt-cond-Sigma-tmp1}0\in (\bSigma^{i+1}-\hat\bSigma_n)_{j\ell}+\lambda\partial |\bSigma^{i+1}_{j\ell}| + \bLambda^i_{j\ell} + (\bSigma^{i+1}-\bTheta^{i+1})_{j\ell}/\mu,\quad \forall j=1,\ldots,p,\ell=1,\ldots,p,\mbox{ and }j\neq \ell,\ee and
\be \label{opt-cond-Sigma-tmp2} (\bSigma^{i+1}-\hat{\bSigma}_n)_{jj}+\bLambda^i_{jj}+(\bSigma^{i+1}-\bTheta^{i+1})_{jj}/\mu=0,\quad \forall j=1,\ldots,p.\ee
Note that by using \eqref{update-Lambda}, \eqref{opt-cond-Sigma-tmp1} and \eqref{opt-cond-Sigma-tmp2} can be respectively rewritten as:
\be\label{opt-cond-Sigma}
(-\bLambda^{i+1}-\bSigma^{i+1}+\hat\bSigma_n)_{j\ell}/\lambda\in\partial |\bSigma^{i+1}_{j\ell}|, \quad \forall j=1,\ldots,p,\ell=1,\ldots,p,\mbox{ and }j\neq \ell,\ee and
\be \label{opt-cond-Sigma-1.5} (\bSigma^{i+1}-\hat{\bSigma}_n)_{jj}+\bLambda^{i+1}_{jj}=0,\quad \forall j=1,\ldots,p.\ee
Using the fact that $\partial |\cdot|$ is a monotone function, \eqref{KKT-1}, \eqref{KKT-1.5}, \eqref{opt-cond-Sigma} and \eqref{opt-cond-Sigma-1.5} imply
\be\label{proof-lemma-sigma-inequa-1}
\langle \bSigma^{i+1}-\hat\bSigma^+, (\hat\bLambda^+-\bLambda^{i+1})+(\hat\bSigma^+-\bSigma^{i+1})\rangle\geq 0.
\ee
The summation of \eqref{proof-lemma-theta-inequa-3} and \eqref{proof-lemma-sigma-inequa-1} gives
\be\label{proof-lemma-combine-inequa-1}
\begin{array}{l}
\langle \bSigma^{i+1}-\hat\bSigma^+, \hat\bLambda^+-\bLambda^{i+1}\rangle +
\langle \hat\bTheta^+-\bTheta^{i+1},\hat\bLambda^+-\bLambda^{i+1} \rangle\\
+\langle \hat\bTheta^+-\bTheta^{i+1}, \bSigma^{i+1}-\bSigma^i \rangle/\mu
\geq \|\bSigma^{i+1}-\hat\bSigma^+\|_F^2.
\end{array}
\ee
Combining \eqref{proof-lemma-combine-inequa-1} with $\bTheta^{i+1} = \mu(\bLambda^i-\bLambda^{i+1})+\bSigma^{i+1}$ and $\hat\bTheta^+=\hat\bSigma^+$ leads to
\be \label{proof-lemma-combine-inequa-2}
\begin{array}{l}
\langle \bSigma^{i+1}-\hat\bSigma^+, \hat\bLambda^+-\bLambda^{i+1}\rangle +
\langle \hat\bSigma^+-\bSigma^{i+1}-\mu(\bLambda^i-\bLambda^{i+1}), \hat\bLambda^+-\bLambda^{i+1} \rangle \\
+ \langle \hat\bSigma^+-\bSigma^{i+1}-\mu(\bLambda^i-\bLambda^{i+1}),\bSigma^{i+1}-\bSigma^i\rangle /\mu
\geq \|\bSigma^{i+1}-\hat\bSigma^+\|_F^2.
\end{array}
\ee
Simple algebraic derivation from \eqref{proof-lemma-combine-inequa-2} yields the following inequality:
\be \label{proof-lemma-final-inequa-1}
\begin{array}{l}
\mu \langle \bLambda^{i+1}-\hat\bLambda^+,\bLambda^i-\bLambda^{i+1} \rangle
+ \langle \bSigma^{i+1}-\hat\bSigma^+,\bSigma^i-\bSigma^{i+1}\rangle\mu \\
\geq \|\bSigma^{i+1}-\hat\bSigma^+\|_F^2 - \langle \bLambda^i-\bLambda^{i+1},\bSigma^i-\bSigma^{i+1} \rangle.
\end{array}
\ee
Rearranging the terms on the left hand side of \eqref{proof-lemma-final-inequa-1} using $\hat\bTheta^+-\bTheta^{i+1} = (\hat\bTheta^+-\bTheta^i) + (\bTheta^i-\bTheta^{i+1})$ and $\hat\bSigma^+-\bSigma^{i+1} = (\hat\bSigma^+-\bSigma^i) + (\bSigma^i-\bSigma^{i+1})$, then \eqref{proof-lemma-combine-inequa-1} can be reduced to
\be
\label{proof-lemma-final-inequa-2}
\begin{array}{ll} & \mu\langle\bLambda^i-\hat\bLambda^+,\bLambda^i-\bLambda^{i+1}\rangle + \langle\bSigma^i-\hat\bSigma^+,\bSigma^i-\bSigma^{i+1}\rangle/\mu \\
\geq & \mu\|\bLambda^i-\bLambda^{i+1}\|_F^2 + \|\bSigma^i-\bSigma^{i+1}\|_F^2/\mu
+\|\bSigma^{i+1}-\hat\bSigma^+\|_F^2-\langle\bLambda^i-\bLambda^{i+1},\bSigma^i-\bSigma^{i+1}\rangle.
\end{array}
\ee

Using the notation of $U^i$ and $U^*$, \eqref{proof-lemma-final-inequa-2} can be rewritten as
\be\label{proof-lemma-final-inequa-3}
\langle U^i-U^*,U^i - U^{i+1} \rangle_G
\geq
\|U^i-U^{i+1}\|_G^2+\|\bSigma^{i+1}-\hat\bSigma^+
\|_F^2-\langle\bLambda^i-\bLambda^{i+1},\bSigma^i-\bSigma^{i+1}\rangle.
\ee
Combining \eqref{proof-lemma-final-inequa-3} with the following identity \[
\|U^{i+1}-U^*\|_G^2=\|U^{i+1}-U^i\|_G^2-2\langle U^{k}-U^{i+1}, U^i-U^*\rangle_G+\|U^i-U^*\|_G^2,
\]
we get
\be \label{proof-lemma-final-inequa-4}
\begin{array}{ll}
& \|U^i-U^*\|_G^2 - \|U^{i+1}-U^*\|_G^2 \\
= & 2\langle U^i-U^{i+1},U^i-U^*\rangle-\|U^{i+1}-U^i\|_G^2 \\
\geq & 2\|U^i-U^{i+1}\|_G^2+2\|\bSigma^{i+1}-\hat\bSigma^+\|^2-
2\langle\bLambda^i-\bLambda^{i+1},\bSigma^i-\bSigma^{i+1}\rangle-\|U^{i+1}-U^i\|_G^2 \\
= & \|U^i-U^{i+1}\|_G^2+2\|\bSigma^{i+1}-\hat\bSigma^+\|^2-
2\langle\bLambda^i-\bLambda^{i+1},\bSigma^i-\bSigma^{i+1}\rangle.
\end{array}
\ee
Now, using \eqref{opt-cond-Sigma} and \eqref{opt-cond-Sigma-1.5} for $i$ instead of $i+1$, we get,
\be\label{opt-cond-Sigma-k}
(-\bLambda^{i}-\bSigma^{i}+\hat\bSigma_n)_{j\ell}/\lambda\in\partial |\bSigma^{i}_{j\ell}|,
\quad \forall j=1,\ldots,p,\ell=1,\ldots,p,\mbox{ and } j\neq \ell,
\ee and
\be \label{opt-cond-Sigma-k-1.5}
(\bSigma^{i}-\hat{\bSigma}_n)_{jj}+\bLambda^{i}_{jj}=0,\quad \forall j=1,\ldots,p.
\ee
Combining \eqref{opt-cond-Sigma}, \eqref{opt-cond-Sigma-1.5}, \eqref{opt-cond-Sigma-k}, \eqref{opt-cond-Sigma-k-1.5} and using the fact that $\partial |\cdot|$ is a monotone function, we obtain,
\[\langle\bSigma^i-\bSigma^{i+1},\bLambda^{i+1}-\bLambda^i+\bSigma^{i+1}-\bSigma^i\rangle\geq 0,\] which immediately implies,
\be\label{proof-lemma-final-inequa-5}
\langle\bSigma^i-\bSigma^{i+1},\bLambda^{i+1}-\bLambda^i\rangle
\geq
\|\bSigma^{i+1}-\bSigma^i\|_F^2 \geq 0.
\ee
By substituting \eqref{proof-lemma-final-inequa-5} into \eqref{proof-lemma-final-inequa-4}, we get the desired result \eqref{lem:conclusion-eq}.
\end{proof}

\begin{proof}[Proof of Theorem \ref{the:main-convergence}]
From Lemma \ref{lem:Fejer-monotone} we can easily get that
\begin{itemize}
\item [(i)] $\|U^i-U^{i+1}\|_G \rightarrow 0$;
\item [(ii)] $\{U^i\}$ lies in a compact region;
\item [(iii)] $\|U^i-U^*\|_G^2$ is monotonically non-increasing and thus converges.
\end{itemize}
It follows from (i) that $\bLambda^i-\bLambda^{i+1}\rightarrow 0$ and $\bSigma^i-\bSigma^{i+1}\rightarrow 0$. Then \eqref{update-Lambda} implies that
$\bTheta^i-\bTheta^{i+1}\rightarrow 0$ and $\bTheta^i-\bSigma^i\rightarrow 0$. From (ii) we obtain that, $U^i$ has a subsequence $\{U^{i_j}\}$ that
converges to $\bar{U}=(\bar\bLambda,\bar\bSigma)$, i.e., $\bLambda^{i_j}\rightarrow\bar\bLambda$ and $\bSigma^{i_j}\rightarrow\bar\bSigma$.
From $\bTheta^i-\bSigma^i\rightarrow 0$ we also get that $\bTheta^{i_j}\rightarrow\bar\bTheta:=\bar\bSigma$. Therefore,
$(\bar\bTheta,\bar\bSigma,\bar\bLambda)$ is a limit point of $\{(\bTheta^i,\bSigma^i,\bLambda^i)\}$.

Note that \eqref{opt-cond-Sigma} and \eqref{opt-cond-Sigma-tmp2} respectively imply that
\be\label{proof-theorem-kkt-1}
(-\bar\bLambda-\bar\bSigma+\hat\bSigma_n)_{j\ell}/\lambda\in\partial |\bar\bSigma_{j\ell}|,\quad\forall j=1,\ldots,p,\ell=1,\ldots,p,\mbox{ and } j\neq \ell,
\ee
and
\be\label{proof-theorem-kkt-1.5}
(\bar\bSigma-\hat{\bSigma}_n)_{jj}+\bar{\bLambda}_{jj} = 0,\quad \forall j=1,\ldots,p,
\ee
and \eqref{opt-cond-theta-1} implies that
\be\label{proof-theorem-kkt-2}
\langle\bar\bLambda,\bTheta-\bar\bTheta\rangle\leq 0,\quad\forall\bTheta\succeq\epsilon\bI.
\ee
\eqref{proof-theorem-kkt-1}, \eqref{proof-theorem-kkt-1.5} and \eqref{proof-theorem-kkt-2} together with $\bar\bTheta=\bar\bSigma$ mean that $(\bar\bTheta,\bar\bSigma,\bar\bLambda)$ is an optimal solution to \eqref{soft-thresholding.alm}. Therefore, we showed that any limit point
of $\{(\bTheta^i,\bSigma^i,\bLambda^i)\}$ is an optimal solution to \eqref{soft-thresholding.alm}.
\end{proof}

\begin{proof}[Proof of Theorem \ref{theorem:stat-convergence}]
Without loss of generality, we may always assume that $E(X_{ij})=0$ for all $1\le i\le n, 1\le j \le p$. By the condition that $\bSigma^0$ is positive definite, we can always choose some very small $\epsilon>0$ such that $\epsilon$ is smaller than the minimal eigenvalue of $\bSigma^0$. We introduce $\bDelta=\bSigma-\bSigma^0$, and then we can write \eqref{soft-thresholding.plus} in terms of $\bDelta$ as follows,
\ben
\hat\bDelta=\arg\min_{\bDelta:\bDelta=\bDelta^\T,\bDelta+\bSigma^0\succeq \epsilon\bI}~\frac{1}2\|\bDelta+\bSigma^0-\hat\bSigma_n\|_F^2+\lambda|\bDelta+\bSigma^0|_1\quad
(\equiv F(\bDelta)).
\een
Note that it is easy to see that $\hat\bDelta=\hat\bSigma^+-\bSigma^0$.

Now we consider $\bDelta\in\{\bDelta:\bDelta=\bDelta^\T,\bDelta+\bSigma^0\succeq \epsilon\bI, \|\bDelta\|_F=5\lambda s^{1/2}\}$. Under the probability event $\{|\hat\sigma^n_{ij}-\sigma^0_{ij}|\le\lambda,~\forall (i,j)\}$, we have
\bean
F(\bDelta)-F(\bzero)
&=&
\frac{1}2\|\bDelta+\bSigma^0-\hat\bSigma_n\|_F^2-\frac{1}2\|\bSigma^0-\hat\bSigma_n\|_F^2
+\lambda|\bDelta+\bSigma^0|_1-\lambda|\bSigma^0|_1\\
&=&
\frac{1}2\|\bDelta\|_F^2+<\bDelta,\bSigma^0-\hat\bSigma_n> +
\lambda|\bDelta_{A_0^c}|_1 + \lambda(|\bDelta_{A_0}+\bSigma^0_{A_0}|_1-|\bSigma^0_{A_0}|_1)\\
&\ge&
\frac{1}2\|\bDelta\|_F^2 - \lambda {(|\bDelta|_1+\sum_i|\Delta_{ii}|)} + \lambda|\bDelta_{A_0^c}|_1 -\lambda|\bDelta_{A_0}|_1\\
&\ge&
\frac{1}2\|\bDelta\|_F^2 - 2\lambda{ (|\bDelta_{A_0}|_1+\sum_i|\Delta_{ii}|)}\\
&\ge&
\frac{1}2\|\bDelta\|_F^2 - 2\lambda { (s+p)}^{1/2}\|\bDelta\|_F\\
&\ge&
\frac{5}2\lambda^2 { (s+p)}\\
&>&
0
\eean

Note that $\hat\bDelta$ is also the optimal solution to the following convex optimization problem
\ben
\hat\bDelta=\arg\min_{\bDelta:\bDelta=\bDelta^\T,\bDelta+\bSigma^0\succeq \epsilon\bI}~F(\bDelta)-F(\bzero)\quad (\equiv G(\bDelta)).
\een
Under the same probability event, $\|\hat\bDelta\|_F\le 5\lambda { (s+p)}^{1/2}$ would always hold.
Otherwise, the fact that $G(\bDelta)>0$ for $\|\bDelta\|_F=5\lambda { (s+p)}^{1/2}$ should contradict with the convexity of $G(\cdot)$ and $G(\hat\bDelta)\le G(\bzero)=0$. Therefore, we can obtain the following probability bound
$$
\Pr(\|\hat\bSigma^+-\bSigma^0\|_F\le 5\lambda { (s+p)}^{1/2})
\ge
1- \Pr(\max_{i,j}|\hat\sigma^n_{ij}-\sigma^0_{ij}|>\lambda).
$$

Now we shall prove the probability bound under the exponential-tail condition. First it is easy to verify two simple inequalities that $1+u\le \exp(u)\le 1+u+\frac{1}2u^2\exp(|u|)$ and $v^2\exp(|v|)\le \exp(v^2+1)$. The first inequality can be proved by using the Taylor expansion, and the second one can be easily derived using the obvious facts that $\exp(v^2+1)\ge\exp(2|v|)$ and $\exp(|v|)\ge v^2$.

Let $t_0=(\eta\frac{\log p}n)^{1/2}$, $c_0=\frac{1}2eK_1\eta^{1/2}+\eta^{-1/2}(M+1)$ and $\varepsilon_0=c_0(\frac{\log p}n)^{1/2}$. For any $M>0$, we can apply the Markov inequality to obtain that
\bean
\Pr(\sum_{i}X_{ij}> n\varepsilon_0)
&\le&
\exp(-t_0n\varepsilon_0)\cdot\prod_{i=1}^nE[\exp(t_0X_{ij})]\\
&\le&
\exp(-t_0n\varepsilon_0)\cdot\prod_{i=1}^n \left\{1+\frac{t_0^2}2E[X_{ij}^2\exp(t_0|X_{ij}|)]\right\}\\
&\le&
p^{-c_0\eta^{1/2}}\cdot\exp(\frac{t_0^2}2\sum_{i=1}^nE[X_{ij}^2\exp(t_0|X_{ij}|)])\\
&\le&
p^{-c_0\eta^{1/2}}\cdot\exp(\frac{t_0^2}2\sum_{i=1}^nE[\exp(t_0^2X_{ij}^2+1)])\\
&\le&
p^{-c_0\eta^{1/2}}\cdot\exp(\frac{1}2eK_1\eta\log p)\quad (=p^{-M-1}),
\eean
where we apply $\exp(u)\le 1+u+\frac{1}2u^2\exp(|u|)$ in the second inequality and $1+u\le\exp(u)$ in the third inequality, and then use $v^2\exp(|v|)\le \exp(v^2+1)$ in the fourth inequality. Moreover, the simple facts that $E[X_{ij}]=0$ ($1\le i\le n$) and $t_0^2=\eta\frac{\log p}n\le\eta$ are also used.

Let $t_1=\frac{1}2\eta(\frac{\log p}n)^{1/2}$ and $c_1=2K_1(\eta^{-1}+\frac{1}4\eta\sigma_{\max}^2)\exp(\frac{1}2\eta\sigma_{\max})+2\eta^{-1}(M+2)$. Define $\varepsilon_1=c_1(\frac{\log p}n)^{1/2}$. For any $M>0$, we first apply the Cauchy inequality to obtain that
\bean
E[X_{ij}^2X_{ik}^2\cdot\exp(\frac{1}2\eta|X_{ij}X_{ik}|)]
&\le&
E[X_{ij}^2X_{ik}^2\cdot\exp(\frac{1}4\eta(X^2_{ij}+X^2_{ik}))]\\
&\le&
(E[X_{ij}^4\exp(\eta X^2_{ij}/2)])^{1/2}\cdot
(E[X_{ik}^4\exp(\eta X^2_{ik}/2)])^{1/2}\\
&\le&
4\eta^{-2}\cdot(E[\exp(\eta X^2_{ij})])^{1/2}
\cdot(E[\exp(\eta X^2_{ik})])^{1/2}\\
&\le&
4K_1\eta^{-2},
\eean
where we use the simple inequality $\exp(|v|)\ge v^2$ in the third inequality. Then, combining this result with the Cauchy inequality again yields that
\bean
&&E[(X_{ij}X_{ik}-\sigma^0_{jk})^2\cdot\exp(t_1|X_{ij}X_{ik}-\sigma^0_{jk}|)]\\
&\le&
2E[X_{ij}^2X_{ik}^2\cdot\exp(\frac{1}2\eta|X_{ij}X_{ik}-\sigma^0_{jk}|)]+
2(\sigma^0_{jk})^2\cdot E[\exp(\frac{1}2\eta|X_{ij}X_{ik}-\sigma^0_{jk}|)]\\
&\le&
8K_1\eta^{-2}\cdot\exp(\frac{1}2\eta\sigma^0_{jk})
+2(\sigma^0_{jk})^2\cdot\exp(\frac{1}2\eta\sigma^0_{jk})\cdot E[\exp(\frac{1}4\eta (X_{ij}^2+X_{ik}^2))]\\
&\le&
8K_1\eta^{-2}\cdot\exp(\frac{1}2\eta\sigma_{\max})+2\sigma_{\max}^2\cdot\exp(\frac{1}2\eta\sigma_{\max})
\bigl(E[\exp(\frac{1}2\eta X_{ij}^2)]\bigr)^{1/2}\cdot\bigl(E[\exp(\frac{1}2\eta X_{ik}^2)]\bigr)^{1/2}\\
&\le&
2K_1(4\eta^{-2}+\sigma_{\max}^2)\cdot\exp(\frac{1}2\eta\sigma_{\max})
\eean
where we use the fact that $t_1=\frac{1}2\eta(\frac{\log p}n)^{1/2}\le\frac{1}2\eta <\eta$ in the first inequality, and then use $|\sigma^0_{jk}|\le (\sigma^0_{jj}\sigma^0_{kk})^{1/2}\le\sigma_{\max}$ in the third inequality. Now, we can apply the Markov inequality to obtain the following probability bound
\bean
&&\Pr(\sum_{i}\{X_{ij}X_{ik}-\sigma^0_{jk}\}> n\varepsilon_1)\\
&\le&
\exp(-t_1n\varepsilon_1)\cdot
\prod_{i=1}^nE\left[\exp(t_1(X_{ij}X_{ik}-\sigma^0_{jk}))\right]\\
&\le&
p^{-\frac{1}2c_1\eta}\cdot\prod_{i=1}^n \left\{1+\frac{1}2t_1^2\cdot
E\left[(X_{ij}X_{ik}-\sigma^0_{jk})^2\cdot\exp(t_1|X_{ij}X_{ik}-\sigma^0_{jk}|)\right]\right\}\\
&\le&
p^{-\frac{1}2c_1\eta} \cdot \exp\left(\frac{1}2t_1^2\cdot \sum_{i=1}^nE\left[(X_{ij}X_{ik}-\sigma^0_{jk})^2\cdot\exp(t_1|X_{ij}X_{ik}-\sigma^0_{jk}|)\right]\right)\\
&\le&
p^{-\frac{1}2c_1\eta}\cdot
\exp\left(K_1(1+\frac{1}4\eta^2\sigma_{\max}^2)\cdot \exp(\frac{1}2\eta\sigma_{\max})\cdot\log p\right)\quad (=p^{-M-2}),
\eean
where we apply $\exp(u)\le 1+u+\frac{1}2u^2\exp(|u|)$ and $E[X_{ij}X_{ik}]=\sigma^0_{jk}$ for $i=1,2,\cdots,n$ in the second inequality, and we use $1+u\le\exp(u)$ in the third inequality.

Recall that $\lambda=c_0\frac{\log p}n+c_1(\frac{\log p}n)^{1/2}=\varepsilon^2_0+\varepsilon_1$ and \[
\hat\sigma^n_{jk}-\sigma^0_{jk}=(\frac{1}n\sum_{i}X_{ij}X_{ik}-\sigma^0_{jk})
-(\frac{1}n\sum_{i}X_{jk})\cdot(\frac{1}n\sum_{i}X_{ik}).
\]
Therefore, we can complete the probability bound under the exponential-tail condition as follows
\bean
\Pr(\max_{j,k}|\hat\sigma^n_{jk}-\sigma^0_{jk}|>\lambda)
&\le&
p^2\Pr(\sum_{i}X_{ij}X_{ik}> n(\sigma^0_{jk}+\varepsilon_1))+2p\Pr(\sum_{i}X_{ij}>n\varepsilon_0)\\
&\le&
3p^{-M}.
\eean

In the sequel we shall prove the probability bound under the polymonial-tail condition. First, we define $c_2=8(K_2+1)(M+1)$ and $\varepsilon_2=c_2(\frac{\log p}n)^{1/2}$. Define $\delta_n=n^{1/4}(\log n)^{-1/2}$, $Y_{ij}=X_{ij}I_{\{|X_{ij}|\le\delta_n\}}$ and $Z_{ij}=X_{ij}I_{\{|X_{ij}|>\delta_n\}}$. Then we have $X_{ij}=Y_{ij}+Z_{ij}$ and $E[X_{ij}]=E[Y_{ij}]+E[Z_{ij}]$. By construction, $|Y_{ij}|\le \delta_n$ are bounded random variables, and { $E[Z_{ij}]$ are bounded by $o(\varepsilon_2)$ due to the fact that
$
|E[Z_{ij}]|\le
\delta_n^{-3}E[{|X_{ij}|^4}I_{\{|X_{ij}|>\delta_n\}}]\le
K_2\delta_n^{-3}=o(\varepsilon_2).
$
} Now we can apply the Bernstein's inequality \citep{bernstein1946,bennett1962} to obtain that
\bean
\Pr(\sum_{i}\{Y_{ij}-E[Y_{ij}])\}> \frac{1}2n\varepsilon_2)
&\le&
\exp\left(\frac{-n\varepsilon_2^2}{8var(Y_{ij})+\frac{4}3\delta_n\varepsilon_2}\right)\\
&\le&
\exp\left(\frac{-c_2\log p}{8K_2+8+O(n^{-1/4})}\right)\\
&=&
O(p^{-M-1}),
\eean
where the fact that $var(Y_{ij})\le E[X^2_{ij}]\le E[X^2_{ij}I_{\{|X_{ij}|\ge 1\}}]+E[X^2_{ij}I_{\{|X_{ij}|\le 1\}}]\le K_2+1$ is used in the second inequality. Besides, we can apply the Markov inequality to obtain that
\[
\Pr(|X_{ij}|>\delta_n)
\le
\delta_n^{-4(1+\gamma+\varepsilon)}E\left[|X_{ij}|^{4(1+\gamma+\varepsilon)}\right]
\le
K_2(\log n)^{2(1+\gamma+\varepsilon)}n^{-1-\gamma-\varepsilon}.
\]
Then, we can derive the following probability bound
\bean
\Pr(\sum_{i}X_{ij}> n\varepsilon_2)
&=&
\Pr(\sum_{i}\{Y_{ij}+Z_{ij}-E[Y_{ij}+Z_{ij}]\}> n\varepsilon_2)\\
&\le&
\Pr(\sum_{i}\{Y_{ij}-E[Y_{ij}]\}> \frac{1}2n\varepsilon_2)+
\Pr(\sum_{i}\{Z_{ij}-E[Z_{ij}]\}> \frac{1}2n\varepsilon_2)\\
&\le&
{ O(p^{-M-1})+\Pr(\sum_{i}\{Z_{ij}-o(\varepsilon_2)\}> \frac{1}2n\varepsilon_2)}\\
&\le&
O(p^{-M-1})+\sum_{i}\Pr(|X_{ij}|>\delta_n)\\
&\le&
O(p^{-M-1})+K_2(\log n)^{2(1+\gamma+\varepsilon)}n^{-\gamma-\varepsilon}.
\eean

Let $c_3=8(K_2+1)(M+2)$ and $\varepsilon_3=c_3(\frac{\log p}n)^{1/2}$. Recall that $\delta_n=(\frac{n}{\log(n)})^{1/4}$, and define $R_{ijk}=X_{ij}X_{ik}I_{\{|X_{ij}|>\delta_n\textrm{~or~}|X_{ik}|>\delta_n\}}$. Then we have $X_{ij}X_{ik}=Y_{ij}Y_{ik}+R_{ijk}$ and $\sigma^0_{jk}=E[X_{ij}X_{ik}]=E[Y_{ij}Y_{ik}]+E[R_{ijk}]$. By construction, $|Y_{ij}Y_{ik}|\le \delta_n^2$ are bounded random variables, and { $E[R_{ijk}]$ is bounded by $o(\varepsilon_3)$ due to the fact that
\bean
|E[R_{ijk}]|&\le&
|E[X_{ij}X_{ik}I_{\{|X_{ij}|>\delta_n\}}]|+|E[X_{ij}X_{ik}I_{\{|X_{ik}|>\delta_n\}}]|\\
&\le&
\delta_n^{-2-4\gamma}E[X^{4(1+\gamma)}_{ij}I_{\{|X_{ij}|>\delta_n\}}]\cdot E[X^2_{ik}]+\delta_n^{-2-4\gamma}E[X^{4(1+\gamma)}_{ik}I_{\{|X_{ik}|>\delta_n\}}]\cdot E[X^2_{ij}]\\
&\le& 2 K_2\delta_n^{-2-4\gamma} \quad (=o(\varepsilon_3)).
\eean} Again, we can apply the Bernstein's inequality to obtain that
\bean
\Pr(\sum_{i}\{Y_{ij}Y_{ik}-E[Y_{ij}Y_{ik}]\}> \frac{1}2n\varepsilon_3)
&\le&
\exp\left(\frac{-n\varepsilon_3^2}{8K_2+8+\frac{4}3\delta_n^2\varepsilon_3}\right)\\
&\le&
\exp\left(\frac{-c_3\log p}{8K_2+8+O((\log n)^{-1/2})}\right)\\
&=&
O(p^{-M-2}),
\eean
where the fact that $var(Y_{ij}Y_{ik})\le E[X^2_{ij}X^2_{ik}]\le (E[X^4_{ij}]E[X^4_{ik}])^{1/2}\le K_2+1$ is used.
\bean
&&\Pr(\max_{j,k}|\sum_{i}(X_{ij}X_{ik}-\sigma^0_{jk})|> n\varepsilon_3)\\
&\le&
\Pr(\max_{j,k}|\sum_{i}\{Y_{ij}Y_{ik}-E[Y_{ij}Y_{ik}]\}|> \frac{1}2n\varepsilon_3)+\Pr(\max_{j,k}|\sum_{i}\{R_{ijk}-E[R_{ijk}]\}|> \frac{1}2n\varepsilon_3)\\
&\le&
{
2\sum_{j,k}\Pr(\sum_{i}\{Y_{ij}Y_{ik}-E[Y_{ij}Y_{ik}]\}> \frac{1}2n\varepsilon_3)
+\Pr(\max_{j,k}|\sum_{i}\{R_{ijk}-o(\varepsilon_3)\}|> \frac{1}2n\varepsilon_3)}\\
&\le&
O(p^{-M})+\sum_{i,j}\Pr(|X_{ij}|>\delta_n)\\
&\le&
O(p^{-M})+K_2p(\log n)^{2(1+\gamma+\varepsilon)}n^{-\gamma-\varepsilon}
\eean

Recall that $\lambda=c_2\frac{\log p}n+c_3(\frac{\log p}n)^{1/2}=\varepsilon^2_2+\varepsilon_3$. Therefore, we can prove the desired probability bound under the polynomial-tail condition as follows
\bean
&&\Pr(\max_{j,k}|\hat\sigma^n_{jk}-\sigma^0_{jk}|>\lambda)\\
&\le&
\Pr(\max_{j,k}|\sum_{i}\{X_{ij}X_{ik}-\sigma^0_{jk}\}|> n\varepsilon_3)
+\Pr(\max_{j}|\sum_{i}X_{ij}|> n\varepsilon_2)\\
&\le&
O(p^{-M})+3K_2p(\log n)^{2(1+\gamma+\varepsilon)}n^{-\gamma-\varepsilon}.
\eean
\end{proof}


\begin{thebibliography}{28}
\newcommand{\enquote}[1]{``#1''}
\expandafter\ifx\csname natexlab\endcsname\relax\def\natexlab#1{#1}\fi

\bibitem[{Anderson(1984)}]{anderson1984}
Anderson, T. (1984), \enquote{{An introduction to multivariate statistical
  analysis}}, John Wiley \& Sons New York.

\bibitem[{Beer et~al.(2002)Beer, Kardia, Huang, Giordano, Levin, Misek, Lin,
  Chen, Gharib, Thomas, et~al.}]{beer2002}
Beer, D., Kardia, S., Huang, C., Giordano, T., Levin, A., Misek, D., Lin, L.,
  Chen, G., Gharib, T., Thomas, D., et~al. (2002), \enquote{Gene-expression
  profiles predict survival of patients with lung adenocarcinoma,}
  \textit{Nature Med.}, 8, 816--824.

\bibitem[{Bennett(1962)}]{bennett1962}
Bennett, G. (1962), \enquote{Probability inequalities for the sum of
  independent random variables,} \textit{J. Amer. Statist. Assoc.}, 33--45.

\bibitem[{Bernstein(1946)}]{bernstein1946}
Bernstein, S. (1946), \enquote{The Theory of Probabilities}, Gostekhizdat,
  Moscow.

\bibitem[{Bickel and Levina(2008a)}]{bickel2008a}
Bickel, P. and Levina, E. (2008a), \enquote{Regularized estimation of large
  covariance matrices,} \textit{Ann. Statist.}, 36, 199--227.

\bibitem[{Bickel and Levina(2008b)}]{bickel2008b}
--- (2008b), \enquote{Covariance regularization by thresholding,} \textit{Ann.
  Statist.}, 36, 2577--2604.

\bibitem[{Cai and Liu(2011)}]{cai2011}
Cai, T. and Liu, W. (2011), \enquote{Adaptive thresholding for sparse
  covariance matrix estimation,} \textit{J. Amer. Statist. Assoc.}, 1--13.

\bibitem[{Cai et~al.(2010)Cai, Zhang, and Zhou}]{czz2010}
Cai, T., Zhang, C., and Zhou, H. (2010), \enquote{Optimal rates of convergence
  for covariance matrix estimation,} \textit{Ann. Statist.}, 38, 2118--2144.

\bibitem[{Cai and Zhou(2011a)}]{cz2011}
Cai, T. and Zhou, H. (2011a), \enquote{Minimax estimation of large covariance
  matrices under $\ell_1$-norm,} \textit{Statist. Sinica, to appear}.

\bibitem[{Cai and Zhou(2011b)}]{cz2012}
--- (2011b), \enquote{Optimal rates of convergence for sparse covariance matrix
  estimation,} \textit{Manuscript}.

\bibitem[{Douglas and Rachford(1956)}]{Douglas-Rachford-56}
Douglas, J. and Rachford, H.~H. (1956), \enquote{On the numerical solution of
  the heat conduction problem in 2 and 3 space variables,} \textit{Trans. Amer.
  Math. Soc.}, 82, 421--439.

\bibitem[{Efron(2009)}]{efron2009}
Efron, B. (2009), \enquote{Are a set of microarrays independent of each other?}
  \textit{Ann. Appl. Stat.}, 3, 922--942.

\bibitem[{Efron(2010)}]{efron2010}
--- (2010), \enquote{Large-Scale Inference: Empirical Bayes Methods for
  Estimation, Testing, and Prediction}, Cambridge University Press.

\bibitem[{El~Karoui(2008)}]{karoui2008}
El~Karoui, N. (2008), \enquote{{Operator norm consistent estimation of large
  dimensional sparse covariance matrices},} \textit{Ann. Statist.}, 36,
  2717--2756.

\bibitem[{Fan and Li(2001)}]{scad}
Fan, J. and Li, R. (2001), \enquote{{Variable selection via nonconcave
  penalized likelihood and its oracle properties},} \textit{J. Amer. Statist.
  Assoc.}, 96, 1348--1360.

\bibitem[{Fortin and Glowinski(1983)}]{fortin1983}
Fortin, M. and Glowinski, R. (1983), \enquote{Augmented Lagrangian methods:
  applications to the numerical solution of boundary-value problems(Book)},
  North-Holland. Co.

\bibitem[{Friedman et~al.(2008)Friedman, Hastie, and Tibshirani}]{glasso}
Friedman, J., Hastie, T., and Tibshirani, R. (2008), \enquote{Sparse inverse
  covariance estimation with the graphical lasso,} \textit{Biostat.}, 9, 432.

\bibitem[{Furrer and Bengtsson(2007)}]{furrer2007}
Furrer, R. and Bengtsson, T. (2007), \enquote{{Estimation of high-dimensional
  prior and posterior covariance matrices in Kalman filter variants},}
  \textit{J. Multivariate Anal.}, 98, 227--255.

\bibitem[{Glowinski and Le~Tallec(1989)}]{glowinski1989}
Glowinski, R. and Le~Tallec, P. (1989), \enquote{Augmented Lagrangian and
  operator-splitting methods in nonlinear mechanics}, SIAM, Philadelphia,
  Pennsylvania.

\bibitem[{Johnstone(2001)}]{johnstone2001}
Johnstone, I. (2001), \enquote{{On the distribution of the largest eigenvalue
  in principal components analysis},} \textit{Ann. Statist.}, 295--327.

\bibitem[{Khan et~al.(2001)Khan, Wei, Ringn{\'e}r, Saal, Ladanyi, Westermann,
  Berthold, Schwab, Antonescu, Peterson, et~al.}]{khan2001}
Khan, J., Wei, J., Ringn{\'e}r, M., Saal, L., Ladanyi, M., Westermann, F.,
  Berthold, F., Schwab, M., Antonescu, C., Peterson, C., et~al. (2001),
  \enquote{{Classification and diagnostic prediction of cancers using gene
  expression profiling and artificial neural networks},} \textit{Nature Med.},
  7, 673--679.

\bibitem[{Peaceman and Rachford(1955)}]{Peaceman-Rachford-55}
Peaceman, D.~H. and Rachford, H.~H. (1955), \enquote{The numerical solution of
  parabolic elliptic differential equations,} \textit{SIAM J. Appl. Math.}, 3,
  28--41.

\bibitem[{Rothman(2011)}]{rothman2011}
Rothman, A. (2011), \enquote{Positive definite estimators of large covariance
  matrices,} \textit{Manuscript}.

\bibitem[{Rothman et~al.(2009)Rothman, Levina, and Zhu}]{rothman2009}
Rothman, A., Levina, E., and Zhu, J. (2009), \enquote{Generalized thresholding
  of large covariance matrices,} \textit{J. Amer. Statist. Assoc.}, 104,
  177--186.

\bibitem[{Scheinberg et~al.(2010)Scheinberg, Ma, and Goldfarb}]{ma2010nips}
Scheinberg, K., Ma, S., and Goldfarb, D. (2010), \enquote{Sparse inverse
  covariance selection via alternating linearization methods,} \textit{Adv.
  Neural Inf. Process. Syst.}

\bibitem[{Subramaniana et~al.(2005)Subramaniana, Tamayoa, Moothaa, Mukherjeed,
  Eberta, Gillettea, Paulovichg, Pomeroyh, Goluba, Landera, et~al.}]{gsea2005}
Subramaniana, A., Tamayoa, P., Moothaa, V., Mukherjeed, S., Eberta, B.,
  Gillettea, M., Paulovichg, A., Pomeroyh, S., Goluba, T., Landera, E., et~al.
  (2005), \enquote{Gene set enrichment analysis: A knowledge-based approach for
  interpreting genome-wide expression profiles,} \textit{Proc. Natl. Acad.
  Sci.}, 102, 15545--15550.

\bibitem[{Wu and Pourahmadi(2003)}]{wu2003}
Wu, W. and Pourahmadi, M. (2003), \enquote{{Nonparametric estimation of large
  covariance matrices of longitudinal data},} \textit{Biometrika}, 90,
  831--844.

\bibitem[{Zou(2006)}]{zou2006}
Zou, H. (2006), \enquote{{The adaptive lasso and its oracle properties},}
  \textit{J. Amer. Statist. Assoc.}, 101, 1418--1429.

\end{thebibliography}
\end{document}